\documentclass[journal]{IEEEtran}
\ifCLASSINFOpdf
\else
\fi
\usepackage{graphics,epstopdf,graphicx,amsthm,amsmath,amssymb,mathptmx,braket,colortbl,color,bm,framed,mathrsfs}
\definecolor{myurlcolor}{rgb}{0,0,0.9}

\newcommand{\proj}[1]{| #1\rangle\!\langle #1 |}

\newcommand{\inner}[2]{\langle #1 , #2\rangle}

\DeclareMathOperator{\trace}{Tr}
\newcommand{\Ptr}[2]{\trace_{#1}\Pa{#2}}
\newcommand{\Tr}[1]{\Ptr{}{#1}}

\newcommand{\Pa}[1]{\left[#1\right]}

\newcommand{\norm}[1]{\left\lVert #1 \right\rVert}

\theoremstyle{plain}
\newtheorem{thm}{Theorem}
\newtheorem{lem}[thm]{Lemma}
\newtheorem{prop}[thm]{Proposition}
\newtheorem{cor}[thm]{Corollary}

\newtheorem{con}[thm]{Conjecture}

\newtheorem{Def}[thm]{Definition}
\newtheorem{Rem}[thm]{Remark}

\def\ot{\otimes}
\def\complex{\mathbb{C}}

\newcommand{\CMM}{\mathcal M}

\newcommand{\be}{\begin{equation}}
\newcommand{\ee}{\end{equation}}

\renewcommand{\geq}{\geqslant}
\renewcommand{\leq}{\leqslant}
\renewcommand{\le}{\leqslant}
\begin{document}
\title{Quantum Ruzsa Divergence to Quantify Magic}

\author{Kaifeng~Bu,        Weichen~Gu,
        and~Arthur~Jaffe
\thanks{Kaifeng Bu is with the Department
Department of Mathematics, The Ohio State University, Columbus, Ohio 43210, USA,
and Department of Physics, Harvard University, Cambridge, Massachusetts 02138, USA,
e-mail: bu.115@osu.edu.}% <-this % stops a space
\thanks{Weichen Gu is with Department of Mathematics, The Ohio State University, Columbus, Ohio 43210, USA, 
and Department of Mathematics and Statistics, University of New Hampshire, Durham, New Hampshire 03824, USA,
e-mail: gu.1213@osu.edu.}% <-this % stops a space
\thanks{Arthur Jaffe is with Department of Physics and Mathematics,  Harvard University, Cambridge, Massachusetts 02138, USA,
e-mail: Arthur\_Jaffe@harvard.edu.}
\thanks{Manuscript received March 29, 2024; revised September 30 and December 10, 2024.}}


\markboth{IEEE TRANSACTIONS ON INFORMATION THEORY,~Vol.~XX, No.~X, XXX~2025}%
{Shell \MakeLowercase{\textit{et al.}}: Bare Demo of IEEEtran.cls for IEEE Journals}



\maketitle

\begin{abstract}
In this work, we investigate the behavior of quantum entropy under quantum convolution and its application in quantifying magic. 
We first establish an entropic, quantum central limit theorem (q-CLT), where the rate of convergence is bounded by the magic gap. We also introduce a new quantum divergence based on quantum convolution, called the 
quantum Ruzsa divergence, to study the stabilizer structure of quantum states.  We conjecture a ``convolutional strong subadditivity'' inequality, which leads to the triangle inequality for the quantum Ruzsa divergence. In addition, we propose two new magic measures, the quantum Ruzsa divergence of magic and quantum-doubling constant,  to quantify the amount of magic in quantum states. Finally, by using the quantum convolution, we 
extend the classical, inverse sumset theory to the quantum case. These results shed new insight into the study of the stabilizer and magic states in quantum information theory.\end{abstract}

\begin{IEEEkeywords}
Quantum entropy, quantum convolution, stabilizer states, magic states
\end{IEEEkeywords}

% \fi
\IEEEpeerreviewmaketitle

\section{Introduction}
\IEEEPARstart{R}{uzsa} divergence (or distance) is an important topic in 
additive combinatorics and plays a crucial role in characterizing the structure of subsets and Gaussians. A prominent application of this concept is in the Freiman-Ruzsa inverse sumset theory, which investigates the properties 
 of sets $A$ such that the size (or the Shannon entropy) of $A+A$ is close to that of $A$. 
In this theory, the quantity $|A+A|/|A|$
 and its entropic counterpart $H(\mathcal{X}+\mathcal{X}')-H(\mathcal{X})$
 are utilized to analyze the characteristics of the subset $A$ 
 and the random variable $\mathcal{X}$~\cite{tao2006additive,tao_2010,KontoyiannisIEEE14,MadimanIEEE18,JiangLiuWu19,green2023sumsets,gowers2023conjecture}. Here, $\mathcal{X}, \mathcal{X}'$ represent i.i.d. copies of a random variable.
For example, by Tao's work~\cite{tao_2010}, given i.i.d. copies $\mathcal{X}$ and $\mathcal{X}'$ of a discrete random variable, if 
$H(\mathcal{X}+\mathcal{X}')-H(\mathcal{X})$ is small, then the distribution of $\mathcal{X}$ is close to the uniform distribution on a generalized arithmetic progression.
Recently, a significant advancement in this field is the proof of Marton's conjecture (or the polynomial Freiman–Ruzsa conjecture) by Gowers, Green, Manners, and Tao~\cite{gowers2023conjecture}.
The functionality of Ruzsa distance relies on sumset inequalities, which establish relationships between the size (or Shannon entropy) of the sumset $A+B$ (or random variables $\mathcal{X}$ and $\mathcal{Y}$) and that  of sets $A$ and $B$ (or $\mathcal{X}$ and $\mathcal{Y}$)~\cite{tao_2010,MadimanIEEE18,Ruzsa09,Madiman08,MadimanRandom12}.

Inspired by the classical work, we introduce the quantum Ruzsa divergence to characterize the stabilizer structure of quantum states in this work. 
Stabilizer states were first introduced by Gottesman~\cite{Gottesman97}, and now have many applications 
including quantum error correction codes~\cite{ShorPRA95,Kitaev_toric}, and the classical simulation of quantum circuits, known as Gottesman-Knill theorem~\cite{gottesman1998heisenberg}. 
These applications indicate that nonstabilizer states and circuits are necessary to achieve the quantum computational advantage. 
Later, the
extension of the Gottesman-Knill theorem beyond stabilizer circuits was further studied~\cite{BravyiPRL16,BravyiPRX16,bravyi2019simulation,BeverlandQST20,SeddonPRXQ21, bu2022classical,gao2018efficient,Bu19,UmeshSTOC23,koh2015further}.
The term ``magic'' was introduced by Bravyi and Kitaev~\cite{BravyiPRA05} to express the property that a state is not a stabilizer.
Magic thus serves as a pivotal quantum resource for understanding the 
boundary between classical and quantum computation. In this work, we explore a new method to quantify magic using the quantum Ruzsa divergence.

The fundamental concept to define quantum Ruzsa divergence is 
to use a quantum convolution that the authors recently introduced~\cite{BGJ23a,BGJ23b,BGJ23c,BJ24a}. 
This quantum convolution defined for qudit and for qubit systems is different from the convolutions defined by Audenaert,  Datta, and Ozols~\cite{Audenaert16} and the free convolution introduced by Voiculescu~\cite{voiculescu2016free}. Our convolution for $n$-qudit states depends on a given unitary, which itself is a tensor product of  $2$-qudit unitaries. 
%that Additionally, it can be constructed by a convolutional unitary, consisting of
%a tensor product of $2$-qudit unitary. 
It can be implemented with
a constant-depth quantum circuit with $2$-qudit gates, which may be realizable in the neutral-atom platform~\cite{Evered2023high}. Moreover,
this quantum convolutional framework provides a new method to understand and study stabilizer states.

One consequence of our convolutional framework is a  quantum central limit theorem for discrete-variable (DV) quantum  systems.
This means that repeated quantum convolution with any zero-mean state converges to a stabilizer state. Therefore we identify the set of stabilizer states as  the set of ``quantum-Gaussian'' states. 
There are also other reasons to consider stabilizer states as the discrete Gaussian states.  One is the 
quadratic form in the expression of computational basis, and the other is the extremality of stabilizer states, as discussed in~\cite{Bu2024a}. 

The central limit theorem (CLT) is a fundamental result in probability theory.
Given i.i.d. random variables $\mathcal{X}_i$ with zero mean and finite variance $\sigma$, the normalized 
sum $\bar{\mathcal{X}}_N=\frac{\sum^N_{i=1}\mathcal{X}_i}{\sqrt{N}}$ converges to a Gaussian random variable $\mathcal{X}_G$ with the same mean value and variance.
 The study of the entropic central limit theorem 
has a long history, tracing back to the work of Linnik~\cite{Linnik59}. Here the entropy is  the Shannon
entropy, defined as $H(\bar{\mathcal{X}}_N)=-\int f_{\bar{\mathcal{X}}_N}\log f_{\bar{\mathcal{X}}_N}$,
 where $f_{\bar{\mathcal{X}}_N}$ is the probability density function of $\bar{\mathcal{X}}_N$. Barron later showed that $H(\bar{\mathcal{X}}_N)$ converges to the $H(\mathcal{X}_G)$ as $N\to \infty$, where $\mathcal{X}_G$ is the corresponding Gaussian random variable~\cite{Barron86}. Furthermore, the rate of convergence in the entropic central limit theorem has attracted much attention~\cite{artstein2004JAMS,artstein2004PTRF,Barron04}. 

In the case of continuous-variable (CV) quantum systems,  various central limit theorems with a Gaussian limit
also have an interesting history including Cushen and Hudson \cite{Cushen71}, and related work of Hepp and Lieb \cite{Lieb73,Lieb1973}. 
Many other quantum or noncommutative versions of the central limit theorem appeared later, see~\cite{JiangLiuWu19,Giri78,Goderis89,Matsui02,Cramer10,Jaksic09,Arous13,Michoel04,GoderisPTRT89,JaksicJMP10,Accardi94,Liu16,Hayashi09,CampbellPRA13,BekerCMP21,Carbone22,beigi2023optimal}. 
For the entropic q-CLT in CV systems, the rate of convergence was given for an $m$-mode quantum state under some 
technical assumption~\cite{BekerCMP21}. In our framework, the stabilizer states are 
the quantum Gaussians, which is different from the previous studies. Hence,
we also explore the rate of convergence of the entropic quantum central limit theorem, which converges to a stabilizer state. We summarize our results in the next section.
%\IEEEpubidadjcol

\subsubsection{Summary of main results}
In \S~\ref{sec:q-CLT}, we establish an entropic q-CLT for DV quantum systems.  We show that the quantum relative entropy 
between the $N$-th repetition of our quantum convolution and the mean state converges to zero at 
an exponential rate. This rate is bounded from below
by the ``magic gap'' defined by the state. Note that we can also derive an entropic q-CLT directly 
from the q-CLT in Hilbert-Schmidt norm~\cite{BGJ23a} by using the relationship between the entropy and norms. However, the constants in the bound of the entropic q-CLT in this work are significantly better than those obtained from the previous approach in~\cite{BGJ23a}.

In \S~\ref{sec:RZ}, we introduce a quantum Ruzsa divergence
to study the stabilizer structure of quantum states. 
We prove the basic mathematical properties of quantum Ruzsa divergence in \S~\ref{sec:PRZ}. 
In addition, we introduce a magic measure, called quantum Ruzsa divergence of magic, to quantify 
the magic of quantum states in \S~\ref{sec:mRZ}. 
This measure can serve as a good magic measure based on the fact that quantum Ruzsa divergence can capture the  discrete Gaussian structure. This
is different from the previous ones based on distance measures, like the relative entropy of magic~\cite{Veitch14}.

To further explore the properties of quantum Ruzsa divergence, we conjecture and investigate the convolutional strong subadditivity of quantum entropy in \S~\ref{sec:CSSA}; this property would lead to the triangle inequality for quantum Ruzsa divergence.
The strong subadditivity \eqref{ineq:CSSA} that we propose is different from the strong subadditivity proved by 
Lieb and Ruskai~\cite{LiebJMP73}.  We prove this convolutional strong subadditivity in two special cases: when all the input states are either diagonal or stabilizer states.
We also consider the convolutional subadditivity, which holds for the classical convolution but not for quantum convolution.
However, 
we find that the violation of 
convolutional subadditivity implies the existence of magic in quantum states.

Finally, we use the quantum Ruzsa divergence to introduce the quantum-doubling  constant, and study the quantum inverse sumset theory to characterize the stabilizer states in \S~\ref{sec:QIST}, which may be of independent research interest. We also show that the quantum-doubling constant can serve as a magic measure that does not require the optimization over all stabilizer states.

Our method based on the Gaussian structure can be applied across various quantum computational models, including matchgate circuits~\cite{lyu2024fermionic}.
Previous measures for quantifying magic~\cite{BravyiPRL16,BravyiPRX16,bravyi2019simulation, Bu19,Veitch14,Veitch12mag,HowardPRL17,BuPRA19_stat,LeonePRL22} either involve (1) optimization over all stabilizer states, making them computationally hard, or (2) are overly specialized, applicable only to stabilizer computations.
Hence, our method addresses the shortcomings of previous measures and offers a comprehensive framework for understanding quantum advantage. Therefore, our new measures will provide significant new insights into both resource theory and quantum computation.

\section{Preliminary}
Fix natural numbers $n$ (the number of qudits) and $d$ (the degree of each qudit) and study  the  Hilbert space $\mathcal{H}^{\ot n}$, where $\mathcal{H} = \complex^d$. 
Consider an orthonormal basis $\set{\ket{k}: k\in \mathbb{Z}_d}$ for the Hilbert space $\mathcal{H}$; here $\mathbb{Z}_{d}$ denotes the cyclic group of order $d$. One calls these vectors the \textit{computational basis}.  
The Pauli operators $X$ and $Z$ are defined by 
\be
X: |k\rangle\mapsto |k+1\rangle\;,
\qquad
Z: |k\rangle \mapsto\omega^k_d|k\rangle\;,
\qquad
\forall k\in \mathbb{Z}_d\;, 
\ee
where $\omega_d=\exp(2\pi i /d)$ is the primitive $d$-th root of unity. 
 We restrict $d$ to be prime in order to define our quantum convolution.  

The local Weyl operators (or generalized Pauli operators)
are 
\[
w(p,q)=\zeta^{-pq}\, Z^{p} X^{q}\;,
\quad\text{where}\quad
\zeta = \left\{
\begin{aligned}
&\omega^{(d+1)/2}_{d}\;, &\text{for }d \text{ odd}\hfill\\
&e^{i\pi/2}\;, &\text{for }d=2\hfill
\end{aligned}
\right.\;.
\]

For the $n$-qudit system on $\mathcal{H}^{\otimes n}$, the Weyl operators are defined as
 \begin{eqnarray}\label{WeylOperators}
w(\vec p, \vec q)
=w(p_1, q_1)\ot...\ot w(p_n, q_n)
=w(-\vec p,-\vec q)^{\dagger}\;,
\end{eqnarray}
with $\vec p=(p_1, p_2,..., p_n)\in \mathbb{Z}^n_d$, $\vec q=(q_1,..., q_n)\in \mathbb{Z}^n_d $. 
Let us denote $\mathcal{P}_n$ as the group generated by the Weyl operators and 
phase $\zeta $.
The Weyl operators are an orthonormal basis for the space of linear operators on $\mathcal{H}^{\ot n}$
 with respect to the inner product 
 \begin{align}
      \inner{A}{B}=\frac{1}{d^n}\Tr{A^\dag B}\;.
 \end{align}
Denote $V^n:=\mathbb{Z}^n_d\times \mathbb{Z}^n_d$; this represents the phase space for $n$-qudit systems, as was studied in~\cite{Gross06}. 
Let $D(\mathcal{H}^{\ot n})$ denote the set of  all quantum states on $\mathcal{H}^{\ot n}$, namely positive matrices with unit trace.

\begin{Def}[\bf Characteristic function]
The characteristic function $\Xi(\vec p, \vec q)$ of a state $\rho\in D(\mathcal{H}^{\ot n})$ is the  coefficient of $\rho$ in the Weyl basis,
\begin{equation*}
\Xi_{\rho}(\vec{p},\vec q)
=\Tr{\rho w(-\vec{p},-\vec q)}\;,
\end{equation*}
and
\begin{equation*}
\rho=
\frac{1}{d^n}
\sum_{(\vec{p},\vec q)\in V^n}
\Xi_{\rho}(\vec{p},\vec q)\ w(\vec{p},\vec q)\;.
\end{equation*}
\end{Def}
\iffalse
Any quantum state $\rho$ can be written as a linear combination of the Weyl operators 
$
\rho=\frac{1}{d^n}
\sum_{(\vec{p},\vec q)\in V^n}
\Xi_{\rho}(\vec{p},\vec q)w(\vec{p},\vec q).
$
\fi
The process of taking characteristic functions is the quantum Fourier transform that we consider.
More details about the properties of the characteristic functions can be found in 
\cite{BGJ23b,Gross06,montanaro2010quantum}.
In this work, we also use 
 $\Xi_{\rho}(\vec x)$ with $\vec x=(\vec p, \vec q)\in V^n$ and the expectation
$$
\mathbb{E}_{k_i\in \mathbb{Z}_d}(\ \cdot \ ):=\frac{1}{d}\sum_{k_i\in\mathbb{Z}_d}(\ \cdot \ )\;.
$$

\begin{Def}[\bf Stabilizer states (Equation(7) in~\cite{Gottesman96} and Chapter 3 in~\cite{Gottesman97})]
A pure stabilizer state $\proj{\psi}$ for an $n$-qudit system  is the projection onto a  stabilizer vector $\ket{\psi}$, namely a common unit eigenvector of an abelian (stabilizer) subgroup of the 
Weyl operators of size $d^n$. That is, 
if the generators of the stabilizer group are 
$\set{g_1,...,g_n}_{i\in [n]}$ with $g_i\in \mathcal{P}_n$, then  
the pure state $\ket{\psi}$ satisfying 
\begin{align}
    g_i\ket{\psi}=\ket{\psi}, \forall i\in [n],
\end{align}
is called a pure stabilizer state. 
And the corresponding density matrix $\proj{\psi}$ can be written as
\begin{eqnarray*}
\proj{\psi}=
\Pi^n_{i=1}\mathbb{E}_{k_i\in \mathbb{Z}_d}g^{k_i}_i\;.
\end{eqnarray*}
A general stabilizer state is a mixed state $\rho$ obtained as a convex linear combination of pure stabilizer states; we denote the set of stabilizer states by  STAB.
\end{Def}

\begin{Def}[\bf Minimal stabilizer-projection state]\label{def:MSPS}
A quantum state $\rho$ is a  {minimal stabilizer-projection state} (MSPS) associated with an abelian subgroup generated by $\set{g_i}_{i\in[r]}$ with $g_i\in \mathcal{P}_n$, if it has the following form 
\begin{eqnarray*}
\rho=
\frac{1}{d^{n-r}}\Pi^r_{i=1}\mathbb{E}_{k_i\in \mathbb{Z}_d}g_i^{k_i}\;.
\end{eqnarray*}
\end{Def}

For example, in an $n$-qudit system and with the abelian group  generated by $\set{Z_1,...,Z_{n-1}}$ where $Z_i$ is the Pauli $Z$ operator on the $i$-th qudit,
the states $\set{\frac{1}{d}\proj{\vec j}\ot I}_{\vec j\in\mathbb{Z}^{n-1}_d}$ are MSPSs. 

\begin{Rem}
Note that 
stabilizer states and their mixed-state generalizations are regarded as the discrete analogs of quantum Gaussian states in this work. This characterization stems from several key properties. For instance, the coefficients of pure stabilizer states in their computational basis decomposition exhibit a quadratic form. Specifically, such states on an $n$-qudit system with odd prime $d$ can be expressed as:
$$\ket{\psi}=\frac{1}{\sqrt{d^n}}\sum_{\vec x\in \mathbb{Z}^n_d}\omega^{f_2(\vec x)}_d\xi_{A\vec x=\vec b}\ket{\vec x},$$ where $f_2$ represents a quadratic function and $\xi_{A\vec{x}=\vec{b}}$ is an indicator function that validates the condition $A\vec{x} = \vec{b}$. For the general $d$, the pure stabilizer 
states still have the above quadratic form with some modification (See Theorem 9 in\cite{garcia2014geometry} and Theorem 13 in~\cite{bu2022classical}). 
Due to the quadratic form, another proof of efficient classical simulation of Clifford circuits has been proposed based on the Gauss sum~\cite{Koh2017computing,bu2022classical}. In addition, the stabilizer states (or MSPSs) usually have extremality, such as the result shown in Lemma~\ref{lem:ent_equ}.  Moreover, there are uncertainty principles that are only saturated by stabilizer states ( See Theorem 1 and 2 in~\cite{Bu2024a}). This is a fundamental property satisfied by Gaussian distributions. These observations made us believe that the stabilizer states (or MSPSs)
 can be regarded as quantum Gaussian states on DV quantum systems.

\end{Rem}

\begin{Def}[\bf Clifford unitary]
An  $n$-qudit unitary $U$ is a Clifford unitary, if 
$Uw(\vec x) U^\dag $ is also a Weyl operator up to a phase for any  $\vec x\in V^n$.
\end{Def}
It is easy to see that Clifford unitaries map stabilizer states to stabilizer states.

\subsection{Basic properties of quantum convolution}
Let us first review the basic knowledge in the framework of quantum convolution. 
\begin{Def}[\bf Mean state (MS) (Definition 3 in \cite{BGJ23a})]\label{def:mean_state}
Given an $n$-qudit state $\rho$,  the mean state  $\mathcal{M}(\rho)$ is the 
operator with the characteristic function: 
\begin{align}\label{0109shi6}
\Xi_{\mathcal{M}(\rho)}(\vec x) :=
\left\{
\begin{aligned}
&\Xi_\rho ( \vec x) , && |\Xi_\rho ( \vec x)|=1,\\
& 0 , && |\Xi_\rho (  \vec x)|<1.
\end{aligned}
\right.
\end{align}
The mean state $\mathcal{M}(\rho)$ is a stabilizer state because 
 the support of the characteristic function forms an abelian group.
\end{Def}
In addition, $\mathcal{M}(\rho)$ has a stabilizer group, i.e., the abelian group generated by the
Pauli operator $w(\vec x)$ such that $w(\vec x)\mathcal{M}(\rho) w(\vec x)^\dag=\mathcal{M}(\rho)$. 
For simplicity, we denote it as $G_{\rho}$.

\begin{Def}[\bf Zero-mean state (Definition 23 in~\cite{BGJ23a})]\label{Def:Zero_mean}
A given  $n$-qudit state $\rho$ has zero-mean, if 
%$\mathcal{M}(\rho)$ has mean-value vector $\vec{\mu}_{\mathcal{M}(\rho)}=(0,..., 0)\mod d$, or equivalently 
the characteristic function of $\mathcal{M}(\rho)$ takes values in $\set{0, 1}$. 
\end{Def}
Note that, if $\rho$ is not a zero-mean state, there exists a Weyl operator $w(\vec x)$ such that $w(\vec x)\rho w(\vec x)^\dag$
is a zero-mean state (see Lemma 90 in \cite{BGJ23b}).

\begin{Def}[\bf Magic gap (Definition 6 in~\cite{BGJ23a})]\label{def:ma_gap}
Given an $n$-qudit state $\rho\in\mathcal{D}(\mathcal{H}^{\ot n})$ for any integer $d\geq 2$, the  magic gap  of $\rho$ is 
\begin{eqnarray*}
MG(\rho)=1-\max_{\vec x\in  \text{Supp}(\Xi_{\rho}): |\Xi_{\rho}(\vec x)|\neq 1}|\Xi_{\rho}(\vec x)|\;.
\end{eqnarray*}
If $\set{\vec x\in  \text{Supp}(\Xi_{\rho}): |\Xi_{\rho}(\vec x)|\neq 1}=\emptyset$, define  $MG(\rho)=0$, i.e., there is no gap on the support of the characteristic function.
\end{Def}

\begin{Rem}
We refer to the difference between the largest and second-largest absolute values of the Fourier coefficients as the "magic gap" because it can be used as a measure of magic in quantum computation~(See Proposition 7 in~\cite{BGJ23a}). That is, the magic gap is always positive, and equal to $0$ iff $\rho$ is a stabilizer state, and invariant under Clifford unitary. 
Additionally, magic gap can be used to provide a lower bound on the number of $T$ gates in the quantum computation with  Clifford and $T$ gates, which provides an operational interpretation of magic gap (See Proposition 8 in~\cite{BGJ23a}).
\end{Rem}

\begin{Def}[\bf Quantum convolution (Definition 20 in \cite{BGJ23a})]\label{def:conv}
    Let $s^2+t^2\equiv 1 \mod d$, with $s,t\neq0$, and let $U_{s,t}$ be the unitary
\begin{align}\label{1231shi1}
 U_{s,t} = \sum_{\vec i,\vec j\in \mathbb{Z}^n_d} |s\vec i+t\vec j \rangle \langle \vec i| \otimes |- t\vec i+s\vec j\rangle \langle \vec j|\;,
 \end{align}
 acting on the $2n$-qudit sytems $\mathcal{H}_A\ot \mathcal{H}_B$ with $\mathcal{H}_A=\mathcal{H}_B=\mathcal{H}^{\ot n}$, and
 the vector $| \vec i \rangle = |  i_1 \rangle \otimes \cdots \otimes |  i_n \rangle \in \mathcal{H}^{\otimes n} $.
The convolution  of two $n$-qudit states $\rho$ and $\sigma$ is 
\begin{align}\label{eq:conv_B}
\rho \boxtimes_{s,t} \sigma = \Ptr{B}{ U_{s,t} (\rho \otimes \sigma) U^\dag_{s,t}}.
\end{align}
\end{Def}

For completeness, we list some useful properties of quantum convolution that we may use throughout.
\begin{lem}\label{lem:key_tech}
The quantum convolution $\boxtimes_{s,t}$ satisfies the following properties:
     \begin{enumerate}
    \item {\bf Convolution-multiplication duality (Proposition 81 in \cite{BGJ23b}):} $\Xi_{\rho\boxtimes_{s,t}\sigma}(\vec x)=\Xi_{\rho}(s\vec x)\Xi_{\sigma}(t\vec x)$, for any $\vec x\in V^n$
    \item  {\bf Convolutional stability (Proposition 45 in \cite{BGJ23b}):} If both $\rho$ and $\sigma$ are stabilizer states,  then 
$\rho\boxtimes_{s,t}\sigma$ is a stabilizer state.
    \item {\bf Quantum central limit theorem (Theorem 91 in \cite{BGJ23b}):} The iterated convolution $\boxtimes^N_{s,t}\rho$   of a zero-mean state $\rho$ converges to $\mathcal{M}(\rho)$ as $N\to \infty$, 
    where $ \boxtimes^{N+1}_{s,t}\rho=(\boxtimes^N_{s,t}\rho)\boxtimes_{s,t} \rho$
  and  $\boxtimes^0_{s,t}\rho=\rho$.
\item {\bf Quantum maximal entropy principle (Theorem 17 in \cite{BGJ23b}):} $S(\rho)\leq S(\mathcal{M}(\rho))$.
\item{\bf Commutativity with Clifford unitaries (Lemma 85 in \cite{BGJ23b}): }For any Clifford unitary $U$, there exists some Clifford unitary $V$ such
that $(U\rho U^\dag)\boxtimes_{s,t} (U\sigma U^\dag)=V(\rho\boxtimes_{s,t} \sigma)V^\dag$ for any input states $\rho$ and $\sigma$.
\item{\bf Commutativity with Weyl operators (Proposition 35 in \cite{BGJ23b} ):} For any Weyl operators $w(\vec x)$ and $w(\vec y)$, we have
$\mathcal{E}\left(w(\vec x)\ot w(\vec y)\rho_{AB}w(\vec x)^\dag\ot w(\vec y)^\dag\right)
=w(s\vec x+t\vec y)\mathcal{E}(\rho_{AB})w(s\vec x+t\vec y)^\dag$
where $\mathcal{E}(\rho_{AB})=\Ptr{B}{U_{s,t}\rho_{AB}U^\dag_{s,t}}$. 
\end{enumerate}
\end{lem}
For simplicity we use $\boxtimes$ to denote the quantum convolution $\boxtimes_{s,t}$ 
for any chosen nonzero $s,t$ in this work. Besides, the comparison between our quantum convolution and the bosonic ones 
can be found in Tables 2 and 3 in~\cite{BGJ23a}.

Note that, for qubit systems, i.e., the local dimension $d=2$,  we change  quantum convolution  to the  Definition 7 of~\cite{BGJ23c}, as
there 
is no nontrivial choice of $s,t\in \mathbb{Z}_2$ such that $s^2+t^2\equiv 1\mod 2$ in Definition~\ref{def:conv}. The qubit convolution in \cite{BGJ23c}
also satisfies the properties in the above Lemma \ref{lem:key_tech} as proved in Theorem 18 in~\cite{BGJ23c}.

\subsection{Basic properties of quantum entropy}
Now, let us review some basic properties of quantum entropy and relative entropy.

\begin{Def}[\bf Quantum entropy]\label{def:Qent}
 Given a quantum state $\rho$, 
 the von Neumann entropy is 
 \begin{eqnarray}
     S(\rho):=-\Tr{\rho\log\rho}.
 \end{eqnarray}
 Given a parameter $\alpha\in [0,+\infty]$,
 the quantum R\'enyi entropy is
\begin{eqnarray}
S_{\alpha}(\rho)
:=\frac{1}{1-\alpha}\log \Tr{\rho^\alpha}.
\end{eqnarray}    
\end{Def}
Note that $\lim_{\alpha\to 1}S_{\alpha}(\rho)=S(\rho)$. Also 
$S_{\infty}(\rho)=\lim_{\alpha\to \infty}S_{\alpha}(\rho)=-\log\lambda_{\max}$, where
$\lambda_{\max}$ is the largest eigenvalue of $\rho$.
\begin{Def}[\bf Quantum relative entropy]
The relative entropy of $\rho$ with respect to $\sigma$ is 
\begin{eqnarray}
D(\rho||\sigma):=\Tr{\rho(\log\rho-\log\sigma)}\;.
\end{eqnarray}
Given a parameter $\alpha\in [0,+\infty]$,
 the sandwiched quantum R\'enyi relative entropy $D_{\alpha}$ is 
\begin{eqnarray*}
    D_{\alpha}(\rho||\sigma):=\frac{1}{\alpha-1}\log\Tr{\left(\sigma^{\frac{1-\alpha}{2\alpha}}\rho\sigma^{\frac{1-\alpha}{2\alpha}}\right)^{\alpha}}.
\end{eqnarray*}
\end{Def}

Note that $\lim_{\alpha\to 1}D_{\alpha}(\rho||\sigma)=D(\rho||\sigma)$, and 
$$D_{\infty}(\rho||\sigma)=\lim_{\alpha\to +\infty}D_{\alpha}(\rho||\sigma)= \min \log\set{\lambda:\rho\leq \lambda\sigma}\;.$$
One fundamental result  is that 
the mean state $\mathcal{M}(\rho)$ is the closest MSPS to the given state
$\rho$ by using the quantum relative entropy as a distance measure.

\begin{lem}[Theorem 4 in\cite{BGJ23a} and Theorem 17 in~\cite{BGJ23b}]\label{lem:ent_equ}
    Given an $n$-qudit state $\rho$, we have 
\begin{align*}
\min_{\sigma\in MSPS}D_{\alpha}(\rho||\sigma)=
D_{\alpha}(\rho|| \CMM(\rho)) =S_{\alpha}(\CMM(\rho))-S_{\alpha}(\rho) \;,
\end{align*}
for any $\alpha\geq 1$.
\end{lem}

\section{Entropic q-CLT for DV systems}\label{sec:q-CLT}
First, we give our main results on the entropic q-CLT, where the limit states are stabilizer states (or MSPSs).
Let us denote 
\begin{align}
    \boxtimes^{N+1}\rho=(\boxtimes^N\rho)\boxtimes \rho,~~\text{and} ~~\boxtimes^0\rho=\rho,
\end{align}
where $\boxtimes$ is short for the quantum convolution $\boxtimes_{s,t}$ for any chosen nonzero $s,t$ (see Definition \ref{def:conv}).
We  provide
 an upper bound on the rate of convergence in the 
entropic q-CLT in terms of the magic gap, where the rate of convergence is exponentially small with respect to the number of convolutions. 

\subsection{Entropic q-CLT for qudit systems}
\begin{thm}[\bf Entropic q-CLT]\label{thm:qclt2}
Given an $n$-qudit state $\rho$, the quantum relative entropy of $\boxtimes^N\rho$ with respect to the 
mean state $\mathcal{M}(\boxtimes^N\rho)$ has the following bound,
\begin{align}
      \nonumber   & D(\boxtimes^N\rho||\mathcal{M}(\boxtimes^N\rho))\\
   \nonumber     =&S(\mathcal{M}(\rho))-S(\boxtimes^N\rho)\\
    \leq& \log\left[1+(1-MG(\rho))^{2N}\left(\Tr{\rho^2}R(\rho)-1\right)
    \right]\\
     \nonumber    \leq & (1-MG(\rho))^{2N}\left(\Tr{\rho^2}R(\rho)-1\right)\;,\\
     \nonumber\to& 0, ~\text{as} ~N\to \infty,
\end{align}
where $R(\rho)$ is the rank of the state $\mathcal{M}(\rho)$.
\end{thm}

\begin{proof}
Since $\mathcal{M}(\rho)$ is a MSPS, 
then by 
Proposition 10 in \cite{BGJ23b}, $\mathcal{M}(\rho)$ 
is a projection $P$ onto some subspace $\mathcal{H}_S$ up to some normalization. 
Hence, its rank is the dimension of (hence the trace of)  the projector $P$.
Let us assume that $G_{\rho}$ is the stabilizer group of $\mathcal{M}(\rho)$, 
and assume a generating set of $G_\rho$ is $\{g_1, ..., g_r\}$,
where $g_1,...,g_r$ are Weyl operators up to a phase.
Then $G_\rho$ is abelian, $|G_\rho| = d^r$, and
\[P = \prod_{j=1}^r \left(\mathbb{E}_{l_j=1}^d g_j^{l_j} \right). \]
Hence, $\trace[P] = \frac{d^n}{|G_{\rho}|}$,
i.e., 
$R(\rho)=\frac{d^n}{|G_{\rho}|}$. Therefore $S(\mathcal{M}(\rho))=\log R(\rho)=\log\frac{d^n}{|G_{\rho}|} $. 
For simplicity, let us take
$ \lambda=(1-MG(\rho))^2= \max_{\vec x\in  \text{Supp}(\Xi_{\rho}): |\Xi_{\rho}(\vec x)|\neq 1}|\Xi_{\rho}(\vec x)|^2$.
Hence
\begin{align*}
 &D(\boxtimes^N\rho||\mathcal{M}(\boxtimes^N\rho))\\
 =&S(\mathcal{M}(\rho))-S(\boxtimes^N\rho)\\
 \leq& S(\mathcal{M}(\rho))-S_2(\boxtimes^N\rho)\\
    =&\log \frac{d^n}{|G_{\rho}|}
    +\log \left(\frac{|G_{\rho}|}{d^n}+\frac{1}{d^n}\sum_{\vec x \notin G}|\Xi_{\boxtimes^N\rho}(\vec x)|^2\right)\\
    \leq& \log \frac{d^n}{|G_{\rho}|}
    +\log \left(\frac{|G_{\rho}|}{d^n}+\frac{\lambda^{N}}{d^n}\sum_{\vec x \notin G}|\Xi_{\rho}(\vec x)|^2\right)\\
    =& \log \frac{d^n}{|G_{\rho}|}
    +\log \left[\frac{|G_{\rho}|}{d^n}+\lambda^{N}\left(\Tr{\rho^2}-\frac{|G_{\rho}|}{d^n}\right)\right]\\
    =&\log\left[1+\lambda^{N}\left(\Tr{\rho^2}\frac{d^n}{|G_{\rho}|}-1\right)
    \right],
\end{align*}
where the second line comes from Lemma \ref{lem:ent_equ}, the third line comes from the monotonicity 
of R\'enyi entropy $S(\rho)\geq S_2(\rho)$, the fifth line comes from the definition of $\lambda$ 
and convolution-multiplication duality, i.e.,
\begin{align}
\Xi_{\boxtimes^N\rho}(\vec x) = \Xi_{\rho}(s^N\vec x) \Xi_{\rho}(s^{N-1}t\vec x) \Xi_{\rho}(s^{N-2}t\vec x)\cdots \Xi_{\rho}(st\vec x)\Xi_{\rho}(t\vec x),
\end{align}
and the fourth  and  sixth lines come from the fact that 
\begin{eqnarray*}
\Tr{\rho^2}
   =\frac{1}{d^n}\sum_{\vec x}|\Xi_{\rho}(\vec x)|^2
   =\frac{|G_{\rho}|}{d^n}+\frac{1}{d^n}\sum_{\vec x\notin G_{\rho}}|\Xi_{\rho}(\vec x)|^2.
\end{eqnarray*}   
\end{proof}

By the entropic q-CLT, we have the following q-CLT based on trace distance as a corollary.
\begin{cor}\label{cor:trace_CLT}
Given an $n$-qudit state $\rho$ with zero mean, we have 
\begin{align}
  \nonumber  &\norm{\boxtimes^N\rho-\mathcal{M}(\rho)}_1\\
    \leq& \sqrt{2 \log\left[1+(1-MG(\rho))^{2N}\left(\Tr{\rho^2}R(\rho)-1\right)
    \right]}\\
    \leq& 
    \sqrt{2}(1-MG(\rho))^{N}\sqrt{\left(\Tr{\rho^2}R(\rho)-1\right)},\quad \text{as}~~ N\to \infty.
\end{align}
\end{cor}

\begin{proof}
    This is a direct corollary of the quantum Pinsker inequality 
  $  
        \frac{1}{2}\norm{\rho-\sigma}^2_1
        \leq D(\rho||\sigma)
 $,
    and the Theorem~\ref{thm:qclt2}.
\end{proof}

\begin{Rem}
Now, let us discuss the advantage of our bound on the rate of convergence in the entropic q-CLT and the trace distance version. 
Previously, there was control over the rate of convergence for the $L_2$ norm, as shown in earlier work~\cite{BGJ23a}:
 $\norm{\boxtimes^N\rho-\mathcal{M}(\rho)}\leq (1-MG(\rho))^N\norm{\rho-\mathcal{M}(\rho)}_2$. 
 Using the relationship between norms on $n$-qudit system $\norm{\cdot}_1\leq \sqrt{d^n}\norm{\cdot}_2$,  we can derive an upper bound for the trace distance $$\norm{\boxtimes^N\rho-\mathcal{M}(\rho)}_1\leq \sqrt{d^n}(1-MG(\rho))^N\norm{\rho-\mathcal{M}(\rho)}_2,$$
 where the factor $\sqrt{d^n}$ is independent of the state $\rho$. Compared to this bound, our bound on the trace distance in Corollary~\ref{cor:trace_CLT} will be better, especially
for the state $\rho$ with small $R(\rho)$.

Moreover, quantum entropy is continuous with respect to the trace norm. To the best of our knowledge, the continuity of quantum entropy has the following form~\cite{audenaert2007sharp,petz2007quantum}:
for any $n$-qudit states $\rho$ and $\sigma$,
$|S(\rho)-S(\sigma)|\leq 2\norm{\rho-\sigma}_1\log (d^n-1)+h(\frac{1}{2}\norm{\rho-\sigma}_1)$ for $\frac{1}{2}\norm{\rho-\sigma}_1\leq 1-\frac{1}{d^n}$, where $h(x):=-x\log x-(1-x)\log(1-x)$. Hence, by the continuity of quantum entropy,  we will get some upper bound on the entropic q-CLT with at least an additional factor $n\log d$, which is worse than the results we obtained in  Theorem~\ref{thm:qclt2}.

\end{Rem}

\begin{figure}
    \centering
    \includegraphics[width=8cm]{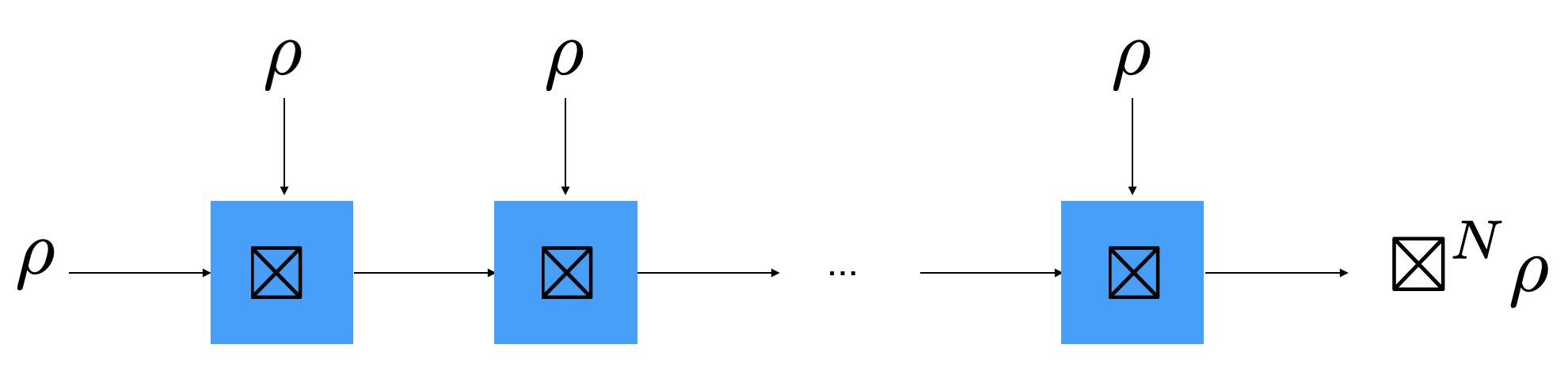}
    \caption{A diagram of the repeated quantum convolution in the q-CLT. }
    \label{fig:enter-label}
\end{figure}

Now, let us generalize the entropic q-CLT from von Neumann entropy to R\'enyi entropy.
\begin{thm}[\bf R\'enyi entropic q-CLT via magic gap]\label{231117thm1}
    Given an $n$-qudit state $\rho$ and any parameter $\alpha\in [1,\infty]$, the $\alpha$-quantum R\'enyi relative entropy of $\boxtimes^N\rho$ with respect to the 
mean state $\mathcal{M}(\boxtimes^N\rho)$ satisfies the  bound,
    \begin{align}
     \nonumber         &D_{\alpha}(\boxtimes^N\rho||\mathcal{M}(\boxtimes^N\rho))\\
    \nonumber          =&S_{\alpha}(\mathcal{M}(\rho))-S_{\alpha}(\boxtimes^N\rho)\\
        \leq& \log\left(1+
(1-MG(\rho))^{N} R(\rho)\sqrt{\Tr{\rho^2}-\frac{1}{R(\rho)}}
        \right)\\
    \nonumber          \leq& (1-MG(\rho))^{N}R(\rho)\sqrt{\Tr{\rho^2}-\frac{1}{R(\rho)}}\\
    \nonumber         \to & 0\;, ~\text{ as}~~ N\to \infty.
    \end{align}
\end{thm}
\begin{proof}
Based on the monotonicity of R\'enyi relative entropy $D_{\alpha}\leq D_{\infty}$ for any $\alpha\geq 0$, we only need to prove the statement for the max-relative entropy
$D_{\infty}$. For simplicity, let us take
$ \lambda=(1-MG(\rho))^2= \max_{\vec x\in  \text{Supp}(\Xi_{\rho}): |\Xi_{\rho}(\vec x)|\neq 1}|\Xi_{\rho}(\vec x)|^2$.
    First, we have 
    \begin{align*}
        &|\bra{\psi}\rho\ket{\psi}|\\
        \leq&  |\bra{\psi}\rho-\mathcal{M}(\rho)\ket{\psi}|
        +|\bra{\psi}\mathcal{M}(\rho)\ket{\psi}|\\
        \leq& \sqrt{\Tr{(\rho-\mathcal{M}(\rho))^2}}
        +\frac{1}{R(\rho)}\\
        =&\sqrt{\frac{1}{d^n}\sum_{\vec x \notin G_{\rho}}|\Xi_{\rho}(\vec x)|^2}
        +\frac{1}{R(\rho)},
    \end{align*}
    where the second line comes from the triangle inequality, the third line comes from
    the Cauchy-Schwarz inequality and the last line comes from the fact that 
    $\rho-\mathcal{M}(\rho)=\frac{1}{d^n}\sum_{\vec x\notin G_{\rho}}\Xi_{\rho}(\vec x)w(\vec x)$.
    Hence, we have 
    \begin{align}
        \nonumber   &\lambda_{\max}(\boxtimes^N\rho)\\
       \nonumber  \leq& \sqrt{\frac{1}{d^n}\sum_{\vec x \notin G_{\rho}}|\Xi_{\boxtimes^N\rho}(\vec x)|^2}
        +\frac{1}{R(\rho)}\\
   \label{ineq:lambda_max}     \leq& \sqrt{\frac{\lambda^{N}}{d^n}\sum_{\vec x \notin G_{\rho}}|\Xi_{\rho}(\vec x)|^2}
        +\frac{1}{R(\rho)}\\
     =&\lambda^{N/2}\sqrt{\Tr{\rho^2}-\frac{1}{R(\rho)}}
       \nonumber  +\frac{1}{R(\rho)},
    \end{align}
 where   the third line comes from the definition of $\lambda$, and the last line comes from the fact that 
\begin{eqnarray*}
\Tr{\rho^2}
   =\frac{1}{d^n}\sum_{\vec x}|\Xi_{\rho}(\vec x)|^2
   =\frac{|G_{\rho}|}{d^n}+\frac{1}{d^n}\sum_{\vec x\notin G_{\rho}}|\Xi_{\rho}(\vec x)|^2.
\end{eqnarray*}   

    Hence, 
    \begin{align*}
        &D_{\infty}(\boxtimes^N\rho||\mathcal{M}(\rho))\\
        =&S_{\infty}(\mathcal{M}(\rho))-S_{\infty}(\boxtimes^N\rho)\\
        =&\log R(\rho)-\log\frac{1}{\lambda_{\max}(\boxtimes^N\rho)}\\
        =&\log \left[R(\rho)\lambda_{\max}(\boxtimes^N\rho)\right]\\
        \leq&\log \left[1+
        \lambda^{N/2}R(\rho)\sqrt{\Tr{\rho^2}-\frac{1}{R(\rho)}}
        \right]\\
         \leq & \lambda^{N/2}R(\rho)\sqrt{\Tr{\rho^2}-\frac{1}{R(\rho)}},\\
        \to &0,~~\text{as}~~ N\to \infty.
    \end{align*}
    where the second line comes from the Lemma~\ref{lem:ent_equ}, 
and the fifth line comes from the inequality~\eqref{ineq:lambda_max}.
\end{proof}

\subsection{Entropic q-CLT for qubit systems}
Now, let us consider the entropic q-CLT for the qubit systems.
In a qubit system, i.e., the local dimension $d=2$, there 
is no nontrivial choice of $s,t\in \mathbb{Z}_2$ such that $s^2+t^2\equiv 1\mod 2$, as required in 
the Definition~\ref{def:conv}. Hence, it is impossible to
consider the discrete beam splitter with two input states using this definition.
 Consequently, we adopt the definition of quantum convolution in Definition 7 of~\cite{BGJ23c}. 
The quantum convolution on an $n$-qubit system under quantum Fourier transform 
(i.e., the characteristic function) also becomes multiplication, which is a common feature shared by 
the one in qudit systems in Definition~\ref{def:conv}. Hence, we will use Definition 7 of~\cite{BGJ23c}
to investigate the entropic q-CLT for the qubit systems.

\begin{Def}[\bf Key Unitary]\label{def:key_U}
The key unitary $V$ for  $K$ quantum systems, with each system containing $n$ qubits, is 
\begin{align}
\nonumber V:=U^{\ot n}
=&U_{1,n+1,...,(K-1)n+1}\ot U_{2,n+2,...,(K-1)n+2}\ot\\
\label{eq:con_cir}& ...\ot U_{n, 2n,..., Kn}.
\end{align}
Here  $U$ is a $K$-qubit unitary constructed using CNOT gates:
\begin{eqnarray}
U:=\left(\prod^K_{j=2}CNOT_{j\to 1}\right)\left(\prod^K_{i=2}CNOT_{1\to i}\right)\;,
\end{eqnarray}
and 
$
CNOT_{2\to 1}\ket{x}\ket{y}=\ket{x+y}\ket{y}
$ for any $x,y\in\mathbb{Z}_2$. 
\end{Def}

\begin{Def}[\bf Convolution of multiple states]\label{Def:conv_qubit}
Given $K$ states $\rho_1,\rho_2,..., \rho_K$, each with $n$-qubits, the multiple convolution $\boxtimes_K$ 
maps $\rho_1,\rho_2,..., \rho_K$  to an $n$-qubit state as follows
\begin{eqnarray}\label{eq:qub_con}
\boxtimes_{K}(\rho_1,\rho_2,...,\rho_K)=\boxtimes_K(\ot^K_{i=1}\rho_i)=\Ptr{1^c}{V\ot^K_{i=1}\rho_i V^\dag}\;.
\end{eqnarray}
Here $V$ is the key unitary in Definition~\ref{def:key_U}, and
$\Ptr{1^c}{\cdot}$ denotes the partial trace taken on the subsystem $2, 3..., K$, i.e., 
 $\Ptr{1^c}{\cdot}=\Ptr{2,3,...,K}{\cdot}$.
\end{Def}

\begin{thm}[\bf Entropic q-CLT for qubits]
Given an $n$-qubit state $\rho$, and $N=2K+1$ for any integer $K\geq 1$, 
\begin{align}
   &\nonumber D(\boxtimes_N\rho||\mathcal{M}(\boxtimes_N\rho))\\
   \nonumber =&S(\mathcal{M}(\rho))-S(\boxtimes_N\rho)\\
    \leq& \log\left[1+(1-MG(\rho))^{2N-2}\left(\Tr{\rho^2}R(\rho)-1\right)
    \right]\\
    \nonumber  \leq& (1-MG(\rho))^{2N-2}\left(\Tr{\rho^2}R(\rho)-1\right)\\
    \nonumber\to& 0,~~\text{as}~~N\to \infty.
\end{align}
\end{thm}

\begin{proof}
The proof is the same as the qudit case.
\end{proof}

\begin{thm}[\bf R\'enyi entropic q-CLT via magic gap]\label{231117thm1}
   Given an $n$-qubit state $\rho$, $N=2K+1$ for any integer $K\geq 1$, and any parameter $\alpha\in [1,\infty]$, we have
    \begin{align}
   \nonumber       &D_{\alpha}(\boxtimes_N\rho||\mathcal{M}(\boxtimes_N\rho))\\
   \nonumber         =& 
S_{\alpha}(\mathcal{M}(\rho))-S_{\alpha}(\boxtimes_N\rho)\\
        \leq& \log\left(1+
(1-MG(\rho))^{N-1} R(\rho)\sqrt{\Tr{\rho^2}-\frac{1}{R(\rho)}}
        \right)\\
  \nonumber        \leq& (1-MG(\rho))^{N-1}R(\rho)\sqrt{\Tr{\rho^2}-\frac{1}{R(\rho)}},\\
  \nonumber        \to& 0, ~~\text{as}~~ N\to \infty.
    \end{align}
\end{thm}
\begin{proof}
    The proof is the same as the qudit case.
\end{proof}

\section{Quantum Ruzsa Divergence}\label{sec:RZ}
In this section, we introduce a new quantum divergence measure derived from quantum convolution, termed the quantum Ruzsa divergence. This measure is motivated by the need in quantum information theory to explore the hidden stabilizer properties of any quantum state. Specifically, for any quantum state, it is essential to determine whether it is a stabilizer state or to quantify how far it deviates from being a stabilizer state—that is, to quantify the amount of magic in the state. Since magic is a crucial quantum resource to achieve quantum computational advantage, its quantification is a fundamental question in the field of quantum computation. We propose that the Ruzsa divergence provides a new family of measures for quantifying magic,  based on the hidden stabilizer structure and its behavior under quantum convolution. Additionally, we expand on this concept by extending the inverse sumset theorem to the qudit case, enabling us to further explore the stabilizer structure of quantum states.

\subsection{Definition and properties}\label{sec:PRZ}

\begin{Def}[\bf Quantum Ruzsa Divergence]\label{Def:QRD}
    Given a quantum convolution $\boxtimes$, the quantum Ruzsa divergence of a state $\rho$ with 
    respect to the state $\sigma$ is
    \begin{eqnarray}
        D_{Rz}(\rho||\sigma):=S(\rho\boxtimes\sigma)-S(\rho).
    \end{eqnarray}
    The $\alpha$-order quantum Ruzsa divergence of $\rho$ with 
    respect to $\sigma$ is
      \begin{eqnarray}
        D_{\alpha,Rz}(\rho||\sigma):=S_{\alpha}(\rho\boxtimes\sigma)-S_{\alpha}(\rho).
    \end{eqnarray}
\end{Def}

\begin{Def}[\bf Symmetrized Quantum Ruzsa Divergence]
    Given two quantum states $\rho$ and $\sigma$, the symmetrized  quantum Ruzsa divergence between 
    $\rho$ and $\sigma$ is 
    \begin{eqnarray}\label{eq:sym_Rz_D}
       d_{Rz}(\rho,\sigma):=\frac{1}{2}\left(S(\rho\boxtimes\sigma)+S(\sigma\boxtimes\rho)-S(\rho)-S(\sigma)\right).
    \end{eqnarray}
    The $\alpha$-order quantum Ruzsa divergence between 
    $\rho$ and $\sigma$ is 
      \begin{align}
          d_{\alpha,Rz}(\rho,\sigma):=\frac{1}{2}\left(S_{\alpha}(\rho\boxtimes\sigma)+S_{\alpha}(\sigma\boxtimes\rho)-S_{\alpha}(\rho)-S_{\alpha}(\sigma)\right).
      \end{align}

\end{Def}

\begin{prop}\label{thm:QRD}
    The quantum Ruzsa divergence satisfies the following properties: 

    (1) {\bf Positivity:} $D_{Rz}(\rho||\sigma)\geq 0$. Also $D_{Rz}(\rho||\sigma)=0$ iff the state $\rho$  is in the abelian C*-algebra generated by stabilizer group $S$ of $\sigma$, i.e.,
$
\rho 
$ is a convex sum of MSPSs associated with $S$. From this we infer that $D_{Rz}(\rho||\rho)=0$ iff $\rho$ is an MSPS.

    (2) {\bf Additivity under tensor product:} $D_{Rz}(\rho_1\ot\rho_2||\sigma_1\ot \sigma_2)=D_{Rz}(\rho_1||\sigma_1)+D_{Rz}(\rho_2||\sigma_2)$.

    (3) {\bf Invariance under Clifford unitary:} $D_{Rz}(U\rho U^\dag|| U\sigma U^\dag)=D_{Rz}(\rho||\sigma)$ for any Clifford unitary $U$.

    (4) {\bf Monotonicity under partial trace:} $D_{Rz}(\Ptr{i}{\rho}||\Ptr{i}{\sigma})\leq D_{Rz}(\rho||\sigma)$, where $\Ptr{i}{\cdot}$ denotes the 
    partial trace on the $i$-th qudit for any $i\in [n]$.

(5)   {\bf Convexity in the first term and concavity in the second:}
    $D_{Rz}(\sum_ip_i\rho_i||\sigma)\leq \sum_ip_iD_{Rz}(\rho_i||\sigma)$, and 
    $D_{Rz}(\rho||\sum_iq_i\sigma_i)\geq\sum_iq_iD_{Rz}(\rho||\sigma_i)$, where 
    $\set{p_i}_i$ and $\set{q_i}_i$ are  classical probability distributions.
    \end{prop}
\begin{proof}
(1) $D_{Rz}(\rho||\sigma)\geq 0$ comes from the entropy inequality under convolution, i.e., 
$S(\rho\boxtimes\sigma)\geq \max\set{S(\rho), S(\sigma)}$ \cite{BGJ23a,BGJ23b}, and the condition for $D_{Rz}(\rho||\sigma)=0$
is the condition for the equality $S(\rho\boxtimes\sigma)=S(\rho)$. (See Theorem 58 in~\cite{BGJ23b}.) 

(2) This follows  from the fact that $(\rho_1\ot \rho_2)\boxtimes (\sigma_1\ot \sigma_2)=(\rho_1\boxtimes\sigma_1)\ot (\rho_2\boxtimes\sigma_2)$.

(3) This is the commutativity of  Clifford unitaries and quantum convolution; see Lemma 85 in~\cite{BGJ23b}. In other words, there always exists some Clifford unitary $U'$ such that 
$(U\rho U^\dag)\boxtimes (U\sigma U^\dag)=U'(\rho\boxtimes\sigma)U'^\dag$
for Clifford unitary $U$.

(4) 
First, we have 
\begin{align*}
    &S\left((\mathbb{E}_{\vec x}w(\vec x)\rho w(\vec x)^\dag)\boxtimes\sigma\right)-S\left(\mathbb{E}_{\vec x}w(\vec x)\rho w(\vec x)^\dag\right)\\
    \leq& \mathbb{E}_{\vec x} \left[S\left((w(\vec x)\rho w(\vec x)^\dag)\boxtimes\sigma\right)-S\left(w(\vec x)\rho w(\vec x)^\dag\right)\right]\\
    =& \mathbb{E}_{\vec x} \left[S\left(w(s\vec x)(\rho\boxtimes\sigma)w(s\vec x)^\dag\right)-S\left(w(\vec x)\rho w(\vec x)^\dag\right)\right] \\
    =&S(\rho\boxtimes\sigma)-S(\rho).
\end{align*}
Here the second line comes from the joint convexity of the quantum relative entropy 
$D(U_{s,t}\rho\ot\sigma U^\dag_{s,t}||\rho\boxtimes\sigma\ot \frac{I}{d^n})$, i.e., 
\begin{align*}
    D\left(U_{s,t}\mathbb{E}_{\vec x}w(\vec x)\rho w(\vec x)^\dag\ot\sigma U^\dag_{s,t}||\mathbb{E}_{\vec x}w(\vec x)\rho w(\vec x)^\dag\boxtimes\sigma\ot \frac{I}{d^n}\right)\\ 
    \leq \mathbb{E}_{\vec x}D\left(U_{s,t}w(\vec x)\rho w(\vec x)^\dag\ot\sigma U^\dag_{s,t}||w(\vec x)\rho w(\vec x)^\dag\boxtimes\sigma\ot \frac{I}{d^n}\right),
\end{align*}
where $D(U_{s,t}\rho\ot\sigma U^\dag_{s,t}||\rho\boxtimes\sigma\ot \frac{I}{d^n})=S(\rho\boxtimes\sigma)+n\log d-S(\rho)-S(\sigma)$.
And the 
third line is a consequence of  Proposition 41 in~\cite{BGJ23b}.  In fact,
\begin{align*}
&\mathcal{E}\left(w(\vec x)\ot w(\vec y)\rho_{AB}w(\vec x)^\dag\ot w(\vec y)^\dag\right)\\
=&w(s\vec x+t\vec y)\mathcal{E}(\rho_{AB})w(s\vec x+t\vec y)^\dag,
\end{align*}
where $\mathcal{E}(\rho_{AB})=\Ptr{B}{U_{s,t}\rho_{AB}U^\dag_{s,t}}$. 
Consequently,  we only need to prove that 
\begin{align}
\nonumber   &D_{Rz}\left(\Ptr{i}{\rho}||\Ptr{i}{\sigma}\right)\\
   =& S\left((\mathbb{E}_{\vec x}w(\vec x)\rho w(\vec x)^\dag)\boxtimes\sigma\right)-S\left(\mathbb{E}_{\vec x}w(\vec x)\rho w(\vec x)^\dag\right)\;.
\end{align}
To prove this statement, note that  
\begin{align*}
    &\left(\Ptr{i}{\rho}\ot \frac{I_i}{d}\right)\boxtimes \sigma\\
    =&\left(\mathbb{E}_{\vec{x}_i\in V}w(\vec x_i)\rho w(\vec x_i)^\dag\right)\boxtimes\sigma\\
    =&\mathbb{E}_{\vec{x}_i\in V}w(s\vec x_i)(\rho\boxtimes\sigma)w(s\vec x_i)^\dag\\
    =& \Ptr{i}{\rho\boxtimes\sigma}\ot \frac{I_i}{d},
\end{align*}
where the first and last lines come from the fact that $\mathbb{E}_{\vec x_i\in V}w(\vec x_i)(\cdot)w(\vec x_i)^\dag=\Ptr{i}{\cdot}\ot \frac{I_i}{d}$, and the second line comes from the following property (See Proposition 41 in~\cite{BGJ23b})
\begin{align*}
\mathcal{E}(w(\vec x)\ot w(\vec y)\rho_{AB}w(\vec x)^\dag\ot w(\vec y)^\dag)
=w(s\vec x+t\vec y)\rho_{AB}w(s\vec x+t\vec y)^\dag,
\end{align*}
where $\mathcal{E}(\rho_{AB})=\Ptr{B}{U_{s,t}\rho_{AB}U^\dag_{s,t}}$. 
Repeating the above process for $\sigma$, we obtain the following result
\begin{align}
    \left(\Ptr{i}{\rho}\ot \frac{I_i}{d}\right)\boxtimes \sigma
    =\Ptr{i}{\rho\boxtimes\sigma}\ot \frac{I_i}{d}
=\Ptr{i}{\rho}\boxtimes\Ptr{i}{\sigma}\ot \frac{I_i}{d}.
\end{align}
Hence, we have 
\begin{align*}
S\left((\mathbb{E}_{\vec x}w(\vec x)\rho w(\vec x)^\dag)\boxtimes\sigma\right)
=&S\left(\Ptr{i}{\rho}\boxtimes\Ptr{i}{\sigma}\ot \frac{I_i}{d}\right)\\
=&S\left(\Ptr{i}{\rho}\boxtimes\Ptr{i}{\sigma}\right)+\log d,
\end{align*}
and 
\begin{align*}
S\left(\mathbb{E}_{\vec x}w(\vec x)\rho w(\vec x)^\dag\right)=
S\left(\Ptr{i}{\rho}\ot \frac{I_i}{d}\right)
=S\left(\Ptr{i}{\rho}\right)+\log d.
\end{align*}
Thus, we have 
\begin{align*}
&S\left((\mathbb{E}_{\vec x}w(\vec x)\rho w(\vec x)^\dag)\boxtimes\sigma\right)
-S\left(\mathbb{E}_{\vec x}w(\vec x)\rho w(\vec x)^\dag\right)\\
=&D_{Rz}\left(\Ptr{i}{\rho}||\Ptr{i}{\sigma}\right).
\end{align*}
Therefore, we obtain the result.

(5) The convexity of $D_{Rz}(\rho||\sigma)$ with respect to the state $\rho$ comes directly from 
the joint convexity of the quantum relative entropy $D\left(U_{s,t}\rho\ot\sigma U^\dag_{s,t}||\rho\boxtimes\sigma\ot \frac{I}{d^n}\right)$. That is, 
\begin{align*}
    &D\left(\sum_ip_iU_{s,t}\rho_i\ot\sigma U^\dag_{s,t}||\sum_ip_i\rho_i\boxtimes\sigma\ot \frac{I}{d^n}\right)\\
    \leq& \sum_ip_i D\left(U_{s,t}\rho_i\ot\sigma U^\dag_{s,t}||\rho_i\boxtimes\sigma\ot \frac{I}{d^n}\right).
\end{align*}
And, the concavity of $D_{Rz}(\rho||\sigma)$ with respect to the state $\sigma$ comes directly from the 
concavity of the von Neumann entropy $S(\cdot)$.  That is, 
\begin{eqnarray*}
    S(\sum_iq_i\rho\boxtimes\sigma_i)
    \geq \sum_iq_iS(\rho\boxtimes\sigma_i).
\end{eqnarray*}
\end{proof}

Note that the quantum Ruzsa divergence differs significantly from the quantum relative entropy. For instance, the quantum relative entropy for identical states is always zero, i.e., $D(\rho||\rho)=0$, which does not hold for quantum Ruzsa divergence. 
Moreover, we can also generalize the results in Proposition~\ref{thm:QRD} to the symmetrized version. 
\begin{cor}
    The symmetrized quantum Ruzsa divergence $d_{Rz}$ satisfies the properties (1)--(4) in Proposition \ref{thm:QRD}.
\end{cor}
\begin{proof}
Here we provide a proof for the completeness. 
By the fact that symmetrized quantum Ruzsa divergence 
$d_{Rz}(\rho, \sigma)$ can be written as 
$d_{Rz}(\rho, \sigma)=\frac{1}{2}[D_{Rz}(\rho||\sigma)+D_{Rz}(\sigma||\rho)]$, it is straightforward to see that $d_{Rz}$ also satisfies the properties (1)--(3).
Moreover, since $D_{Rz}(\Ptr{i}{\rho}||\Ptr{i}{\sigma})\leq D_{Rz}(\rho||\sigma)$ for any states $\rho$ and $\sigma$, we have 
\begin{align*}
    &d_{Rz}(\Ptr{i}{\rho},\Ptr{i}{\sigma})\\
    =&\frac{1}{2}
    \left[
    D_{Rz}(\Ptr{i}{\rho}||\Ptr{i}{\sigma})+D_{Rz}(\Ptr{i}{\sigma}||\Ptr{i}{\rho})
    \right]\\
    \leq& \frac{1}{2}
     \left[
    D_{Rz}(\rho||\sigma)+D_{Rz}(\sigma||\rho)
    \right]\\
=& d_{Rz}(\rho, \sigma).
\end{align*}
Hence, the symmetrized quantum Ruzsa divergence $d_{Rz}$
also satisfies the monotonicity under partial trace, i.e., the property (4) in Proposition~\ref{thm:QRD}.
\end{proof}
Note that, since $d_{Rz}(\rho, \sigma)=\frac{1}{2}[D_{Rz}(\rho||\sigma)+D_{Rz}(\sigma||\rho)]$, and 
$D_{Rz}$ is convex in the first term and concave in the second, $d_{Rz}$ in general does not have the convexity or concavity.

Moreover,  the properties (1)-(3) also hold for the $\alpha$-order
 quantum Ruzsa divergence for $1\leq \alpha<+\infty$; this is a consequence of the property of quantum R\'enyi entropy 
 and Lemma 57 in~\cite{BGJ23b}.

In addition, we have the following result to 
characterize the stabilizer states by using the symmetrized quantum Ruzsa divergence.

\begin{prop}\label{prop:100_QPFR}
    Let $\psi_1$ and $\psi_2$ be two pure $n$-qudit states for which  $d_{Rz}(\psi_1,\psi_2)=0$. Then 
    there exists a pure stabilizer state $\phi_{stab}$ such that 
    $$d_{Rz}(\psi_1,\phi_{stab})=d_{Rz}(\psi_2,\phi_{stab})=0\;.$$ Here  $d_{Rz}$ is the symmetric Ruzsa divergence defined in \eqref{eq:sym_Rz_D}.
\end{prop}
\begin{proof}
   Based on the definition of $d_{Rz}$ and the assumption that $\psi_1$ and $\psi_2$ are pure states, we infer that  $S(\psi_1\boxtimes\psi_2)=0$, i.e., $\psi_1\boxtimes\psi_2$ is a pure state. 
Hence
\begin{eqnarray}\label{eq:con_l2}
   1= \frac{1}{d^n}\sum_{\vec x\in V^n}|\Xi_{\psi_1\boxtimes\psi_2}(\vec x)|^2
=\frac{1}{d^n}\sum_{\vec x\in V^n}|\Xi_{\psi_1}(s\vec x)|^2|\Xi_{\psi_2}(t\vec x)|^2.
\end{eqnarray}
In addition, 
\begin{eqnarray*}
    1=\frac{1}{d^n}\sum_{\vec x\in V^n}|\Xi_{\psi_1}(\vec x)|^2=\frac{1}{d^n}\sum_{\vec x\in V^n}|\Xi_{\psi_2}(\vec x)|^2\;,
\end{eqnarray*}
and 
\begin{eqnarray*}
    |\Xi_{\psi_1}(\vec x)|\leq 1,  |\Xi_{\psi_2}(\vec x)|\leq 1, \forall \vec x\;.
\end{eqnarray*}
Hence, for any $\vec x$, 
\begin{eqnarray}\label{eq:sam_G}
    |\Xi_{\psi_1}(s\vec x)|^2=|\Xi_{\psi_2}(t\vec x)|^2=0 ~\text{or }~1.
\end{eqnarray}
Thus, the sizes of the supports $\text{Supp}(\Xi_{\psi_1})$ and $\text{Supp}(\Xi_{\psi_2})$ are $d^n$.
This implies that both 
$\psi_1$ and $\psi_2$ are  pure stabilizer states, based on the result that a pure state is a stabilizer state, iff the size of the support of its characteristic function is $d^n$~\cite{Bu19}. Moreover, both $\psi_1$ and $\psi_2$ have the same 
stabilizer group as $\psi_1$ up to some phase, based on \eqref{eq:con_l2}. Therefore, we obtain the result.

\end{proof}

\begin{Rem}
A natural generalization of  
Proposition \ref{prop:100_QPFR} is the following conjecture:  
if $d_{Rz}(\psi_1,\psi_2)\leq \epsilon$, then there exists some pure stabilizer state $\phi_{stab}$ such that both 
    $$d_{Rz}(\psi_1,\phi_{stab})\leq c\epsilon\;, \quad
    \text{and}
    \quad
    d_{Rz}(\psi_2,\phi_{stab})\leq c\epsilon\;,$$ where $c$ is some 
    constant independent of $n$. This question is a quantum generalization of  Theorem 1.8 in \cite{gowers2023conjecture}.
\end{Rem}

\subsection{Magic measure via quantum Ruzsa divergence}\label{sec:mRZ}
Since the quantum Ruzsa divergence can be used to characterize the structure of stabilizer states, here we introduce  a  new magic measure via the quantum Ruzsa  divergence in the resource theory of magic. This theory recognizes that magic is essential for achieving a quantum computational advantage.
\begin{Def}\label{Def:MRZ}
The quantum Ruzsa divergence of magic $M_{Rz}(\rho)$ of a state $\rho$ is:%  $M_{Rz}(\rho)$, namely 
\begin{eqnarray}
M_{Rz}(\rho):=\min_{\sigma\in STAB}
D_{Rz}(\rho||\sigma).
\end{eqnarray}
\end{Def}

Let us first introduce the stabilizer channel which will map stabilizer states to stabilizer states, that is the set of free operations in the resource theory of magic.

\begin{Def}[\bf Stabilizer channel]
        A quantum channel  $\Phi:D(\mathcal{H}_A)\to D(\mathcal{H}_A)$ on an $n$-qudit system  $\mathcal{H}_A=\mathcal{H}^{\ot n}$ is a stabilizer channel if it has the Stinespring representation
    \begin{eqnarray}
        \Phi(\rho)=\Ptr{B}{U(\rho\ot\sigma) U^\dag},
    \end{eqnarray}
    where $\sigma$ is a stabilizer state, and $U:\mathcal{H}_A\ot \mathcal{H}_B\to\mathcal{H}_A\ot\mathcal{H}_B$ is a Clifford unitary.
    
\end{Def}
\begin{thm}
Given an $n$-qudit state $\rho$, the quantum Ruzsa divergence of magic $M_{Rz}(\rho)$ satisfies the following properties:

(1) Faithfulness: $M_{Rz}(\rho)\geq 0$; also for pure state $\rho$,  $M_{Rz}(\rho)=0$, iff $\rho$ is an MSPS. 

(2) Monotonicity under stabilizer channels:
$M_{Rz}(\Phi(\rho)) \le M_{Rz}(\rho)$ for any stabilizer channel $\Phi$.
\end{thm}

\begin{proof}
Property (1) follows  from the positivity in Proposition~\ref{thm:QRD}.

(2) 
To prove the monotonicity of $M_{Rz}$ under a stabilizer channel, we only need to prove the monotonicity of $M_{Rz}$ for three cases:  
(2a) monotonicity under tensor product of a stabilizer state;
(2b) monotonicity under Clifford unitary ; and
(2c) monotonicity under partial trace.

(2a) follows from the fact that additivity under the tensor product in Proposition~\ref{thm:QRD}. 

(2b) comes from the invariance of the quantum Ruzsa divergence under Clifford unitary in Proposition~\ref{thm:QRD},
i.e., $D_{Rz}(U\rho U^\dag||U\sigma U^\dag)=D_{Rz}(\rho||\sigma)$ for any Clifford unitary, 
and the closedness of the stabilizer states under Clifford unitary, i.e., 
$U\sigma U^\dag$ is a stabilizer state, for any Clifford unitary $U$ and stabilizer state $\sigma$.

(2c)
 By the monotonicity of quantum Ruzsa divergence under partial trace  in Proposition~\ref{thm:QRD}, we have 
 \begin{eqnarray*}
     D_{Rz}(\rho||\sigma)\geq D_{Rz}\left(\Ptr{i}{\rho}||\Ptr{i}{\sigma}\right)\geq 
     \min_{\sigma'\in \text{STAB}} D_{Rz}\left(\Ptr{i}{\rho}||\sigma'\right)\;,
 \end{eqnarray*}
 for any $\sigma\in \text{STAB}$, where the last inequality comes from the closedness of stabilizer states under partial trace.

\end{proof}

\begin{Rem}
Note that in previous studies on resource theory, magic measures are typically defined through distance measures, such as the relative entropy of magic~\cite{Veitch14}. These measures quantify the distance between given states and the set of stabilizer states, rather than capturing the structure of stabilizer states. However, the magic measure based on quantum Ruzsa divergence differs from conventional approaches, as it is effective due to its ability to characterize the stabilizer structure (or discrete Gaussian structure).
\end{Rem}

Although the quantum Ruzsa divergence is different from the quantum relative entropy in general, we find that they have the following  relation  by taking minimization over
stabilizer states.
\begin{prop}\label{thm:equiv}
For any quantum state $\rho$, we have 
    \begin{eqnarray}
       M_{Rz}(\rho)\leq\min_{\sigma\in\text{MSPS}}D(\rho||\rho\boxtimes\sigma)\;,
    \end{eqnarray}
    where $D(\rho||\rho\boxtimes\sigma)$ is the quantum relative entropy of $\rho$ with respect to $\rho\boxtimes\sigma$.
\end{prop}

\begin{proof}
First, by the concavity of quantum entropy, the minimization in $ \min_{\sigma\in\text{STAB}}S(\rho\boxtimes\sigma)-S(\rho)$
is taken on pure stabilizer states. Hence, we only need to consider the
case where $\sigma$ is a pure stabilizer state, i.e., $\sigma=\proj{\sigma}$.
Without the loss of generality, we assume that the corresponding stabilizer group  $G_{\sigma}=\langle\set{w(\vec x_1),..., w(\vec x_n)}\rangle$, that is, 
$w(\vec x_i)\ket{\sigma}=\ket{\sigma}$ for any $i\in [n]$. Moreover, the stabilizer group $G_{\sigma}$ will induce an orthonormal basis $B_{G_{\sigma}}=\set{\ket{\sigma_{\vec k}}}_{\vec k\in \mathbb{Z}^n_d}$,
where $w(\vec x_i)\ket{\sigma_{\vec k}}=\omega^{k_i}_d\ket{\sigma_{\vec k}}$ for any $i\in [n]$ with $\vec k=(k_1,...,k_n)\in \mathbb{Z}^n_d$. Here, $\ket{\sigma}=\ket{\sigma_{\vec 0}}$
is one element of this basis. That is, 
the density matrix  can be written as $\proj{\sigma_{\vec k}}=
\Pi^n_{i=1}\mathbb{E}_{l_i=1}^d(\omega^{k_i}_dw(\vec x_i))^{l_i}=\frac{1}{d^n}\sum_{\vec l}
\omega^{\vec k\cdot \vec l}_d
\Pi^n_{i=1}w(l_i\vec x_i)
$.
Then, for any state $\rho$, the average over the stabilizer group $G_{\sigma}$, given by $\mathbb{E}_{w(\vec x) \in G_{\sigma}} w(\vec x) \rho w(\vec x)^\dag$, commutes with each element $w(\vec x) \in G_{\sigma}$. Therefore, this average can be expressed as a convex combination of the common eigenstates of elements in $G_{\sigma}$, that is, the basis $B_{G_{\sigma}}$.
Thus, we have
\begin{eqnarray}
  \mathbb{E}_{w(\vec x)\in G_{\sigma}}
w(\vec x)\rho w(\vec x)^\dag
=\sum_{\vec k}\bra{\sigma_{\vec k}}\rho\ket{\sigma_{\vec k}}
\proj{\sigma_{\vec k}}.
\end{eqnarray}
Let us define $\Delta_{B_{G_{\sigma}}}$ as 
the fully-dephasing channel with respect to basis $B_{G_{\sigma}}$ , i.e., 
$    \Delta_{B_{G_{\sigma}}}(\rho)
    =\sum_{\vec k}\bra{\sigma_{\vec k}}\rho\ket{\sigma_{\vec k}}
\proj{\sigma_{\vec k}}.$
Then 
\begin{align}
    \mathbb{E}_{w(\vec x)\in G_{\sigma}}
w(\vec x)\rho w(\vec x)^\dag=  \Delta_{B_{G_{\sigma}}}(\rho).
\end{align}
Then  $\rho\boxtimes\sigma$ is diagonal in the basis $B_{G_{\sigma}}$.
This is because 
\begin{align*}
\rho\boxtimes\sigma
=&\rho\boxtimes(\mathbb{E}_{w(\vec x)\in G_{\sigma}}w(\vec x)\sigma w(\vec x)^\dag)\\
=&\mathbb{E}_{w(\vec x)\in G_{\sigma}}\rho\boxtimes(w(\vec x)\sigma w(\vec x)^\dag)\\
=&\mathbb{E}_{w(\vec x)\in G_{\sigma}}w(t\vec x)(\rho\boxtimes\sigma) w(t\vec x)^\dag\\
=&\mathbb{E}_{w(\vec x)\in G_{\sigma}}
(w(s^{-1}t\vec x)\rho w(s^{-1}t\vec x)^\dag)\boxtimes\sigma,
\end{align*}
where the first line comes from the fact that $w(\vec x)\sigma w(\vec x)^\dag =\sigma$ for any 
$w(\vec x)\in G_{\sigma}$, and the 
last two lines come from the following property (See Proposition 41 in\cite{BGJ23b})
\begin{align*}
&\mathcal{E}(w(\vec x)\ot w(\vec y)\rho_{AB}w(\vec x)^\dag\ot w(\vec y)^\dag)\\
=&w(s\vec x+t\vec y)\mathcal{E}(\rho_{AB})w(s\vec x+t\vec y)^\dag,
\end{align*}
where $\mathcal{E}(\rho_{AB})=\Ptr{B}{U_{s,t}\rho_{AB}U^\dag_{s,t}}$.
Since $G_{\sigma}$ is an abelian group, then $w(\vec x)\in G_{\sigma}$ implies that $w(s^{-1}t\vec x)\in G_{\sigma}$. Hence,  
$\mathbb{E}_{w(\vec x)\in G_{\sigma}}
w(s^{-1}t\vec x)\rho w(s^{-1}t\vec x)^\dag
=\mathbb{E}_{w(\vec x)\in G_{\sigma}}
w(\vec x)\rho w(\vec x)^\dag$. Then we have 
\begin{align}\label{eq:fulld}
    \rho\boxtimes\sigma
=\Delta_{B_{G_{\sigma}}}(\rho)\boxtimes\sigma
=\Delta_{B_{G_{\sigma}}}(\rho\boxtimes\sigma),
\end{align}
which is diagonal in the basis 
$B_{G_{\sigma}}$.

Since 
$  \Xi_{ \proj{\sigma_{\vec k}}
    \boxtimes\sigma}
    =\Xi_{\proj{\sigma_{\vec k}}}(s\vec x)
    \Xi_{\sigma}(t\vec x)$
and  $  \proj{\sigma_{\vec k}}$ are stabilizer states with respect to the same abelian group $G_{\sigma}$ for different eigenvalues, then 
$\proj{\sigma_{\vec k}}
    \boxtimes \sigma$ is also a stabilizer state with respect to the same abelian group $G_{\sigma}$, i.e., $\proj{\sigma_{\vec k}}
    \boxtimes \sigma\in B_{G_{\sigma}}$. Moreover, 
\begin{eqnarray}
    \proj{\sigma_{\vec k}}
    \boxtimes\proj{\sigma}
    =\proj{\sigma_{s\vec k}},\quad \forall \vec k\in \mathbb{Z}^n_d.
\end{eqnarray}
That is, 
\begin{eqnarray}
    \rho\boxtimes\sigma
    =\Delta_{B_{G_{\sigma}}}(\rho\boxtimes\sigma)=
    \sum_{\vec k}\bra{\sigma_{\vec k}}\rho\ket{\sigma_{\vec k}}
\proj{\sigma_{s\vec k}}.
\end{eqnarray}
Then 
\begin{align*}
    S( \rho\boxtimes\sigma)
    =&S(\sum_{\vec k}\bra{\sigma_{\vec k}}\rho\ket{\sigma_{\vec k}}
\proj{\sigma_{\Pi(\vec k)}})\\
=&-\sum_{\vec k}\bra{\sigma_{\vec k}}\rho\ket{\sigma_{\vec k}}\log \bra{\sigma_{\vec k}}\rho\ket{\sigma_{\vec k}}\\
=&S(\Delta_{B_{\sigma}}(\rho)).
\end{align*}
Thus, 
    \begin{eqnarray}
\min_{\sigma\in\text{STAB}}S(\rho\boxtimes\sigma)-S(\rho)=\min_{B_{STAB}}S(\Delta_{B_{STAB}}(\rho))-S(\rho),
    \end{eqnarray}
    where $\min_{B_{STAB}}$ denotes the minimization over all the orthonormal basis generated by the 
    stabilizer states, and $\Delta_{B_{STAB}}$ is the corresponding fully-dephasing channel.

    Moreover, by using  \eqref{eq:fulld}, we can also get
    \begin{eqnarray*}
        -\Tr{\rho\log\rho\boxtimes\sigma}
        &=&-\Tr{\rho\log\Delta_{B_{G_{\sigma}}}(\rho\boxtimes\sigma)}\\
        &=&-\Tr{\Delta_{B_{G_{\sigma}}}(\rho)\log\Delta_{B_{\sigma}}(\rho\boxtimes\sigma)}\\
        &\geq& -\Tr{\Delta_{B_{G_{\sigma}}}(\rho)\log\Delta_{B_{G_{\sigma}}}(\rho)},
    \end{eqnarray*}
    where the last line comes from the non-negativity of the relative entropy $D(\Delta_{B_{G_{\sigma}}} (\rho)||\Delta_{B_{G_{\sigma}}}(\rho\boxtimes\sigma))$.
    Hence, we have 
$$\min_{\sigma\in\text{MSPS}}D(\rho||\rho\boxtimes\sigma)\geq \min_{B_{STAB}}S(\Delta_{B_{STAB}}(\rho))-S(\rho).$$
This is the desired result.
\end{proof}

\subsection{Convolutional strong  subadditivity}\label{sec:CSSA}
Recall that in the classical setting, when considering three independent random variables $\mathcal{X}, \mathcal{Y}, \mathcal{Z}$ on $\mathbb{Z}_d$, the classical Ruzsa divergence $D_{Rz}(\mathcal{X}||\mathcal{Y}):=H(\mathcal{X}+\mathcal{Y})-H(\mathcal{X})$ satisfies the triangle inequality (see Theorem 1 in~\cite{MadimanIEEE18}): 
\begin{eqnarray}
D_{Rz}(\mathcal{X}||\mathcal{Z}) \leq D_{Rz}(\mathcal{X}||\mathcal{Y}) + D_{Rz}(\mathcal{Y}||\mathcal{Z}), 
\end{eqnarray} 
which can also be expressed as: 
\begin{align}
H(\mathcal{X}+\mathcal{Z}) + H(\mathcal{Y}) \leq H(\mathcal{X}+\mathcal{Y}) + H(\mathcal{Y}+\mathcal{Z}). 
\end{align} Here, $H(\mathcal{X})$ represents the Shannon entropy of the random variable. This inequality is derived from a more general inequality: 
\begin{align} 
H(\mathcal{X}+\mathcal{Y}+\mathcal{Z}) + H(\mathcal{Y}) \leq H(\mathcal{X}+\mathcal{Y}) + H(\mathcal{Y}+\mathcal{Z}), \end{align} 
which follows directly from the data processing inequality: 
\begin{align} 
I(\mathcal{X}:\mathcal{X}+\mathcal{Y}+\mathcal{Z}) \leq I(\mathcal{X}:\mathcal{X}+\mathcal{Y}). 
\end{align} 

% Here, we are also wondering whether the quantum Ruzsa divergence satisfies the 
% tringle inequality. Hence, we have the following conjecture.

Here, we are also interested in whether the quantum Ruzsa divergence satisfies the triangle inequality. Hence, we propose the following conjecture:

\begin{con}\label{Conj:sub_RZ}
Given a balanced convolution $\boxtimes_{s,t}$, i.e., $s\equiv t \mod d$, the quantum Ruzsa divergence 
satisfies  the triangle inequality 
    \begin{eqnarray}
        D_{Rz}(\rho||\tau)\leq D_{Rz}(\rho||\sigma)+D_{Rz}(\sigma||\tau)\;,
    \end{eqnarray}
which is equivalent to
    \begin{eqnarray}\label{ineq:tri_equiv}
        S(\rho\boxtimes\tau)+S(\sigma)
        \leq S(\rho\boxtimes\sigma)+S(\sigma\boxtimes\tau)\;.
    \end{eqnarray}
\end{con}
Note that the inequality~\eqref{ineq:tri_equiv} is not balanced; the state $\sigma$ appears once on the left-hand side but twice on the right-hand side. Hence, we
proposed the following balanced version, which we call "convolutional strong subadditivity".

\begin{con}[\bf Convolutional strong subadditivity]
Given three $n$-qudit quantum states $\rho,\sigma, \tau$, we have 
\begin{eqnarray}\label{ineq:CSSA}
     S(\rho\boxtimes\tau\boxtimes\sigma)+S(\sigma)
        \leq S(\rho\boxtimes\sigma)+S(\sigma\boxtimes\tau)\;,
\end{eqnarray}
where $\rho\boxtimes\tau\boxtimes\sigma:=(\rho\boxtimes_{s,t}\tau)\boxtimes_{l,m}\sigma$, $\rho\boxtimes\sigma:=\rho\boxtimes_{s,t}\sigma$, $\sigma\boxtimes\tau:=\sigma\boxtimes_{s,t}\tau$, $s^2+t^2\equiv 1\mod d$, $s\equiv t\mod d$, $m\equiv ls \mod d$, and $l^2+m^2\equiv 1\mod d$.
\end{con}
The reason that we choose the parameters in this way is
to generalize  the classical balanced convolution on three random variables $\frac{\mathcal{X}+\mathcal{Y}+\mathcal{Z}}{\sqrt{3}}=\sqrt{\frac{2}{3}}\left(\frac{\mathcal{X}+\mathcal{Y}}{\sqrt{2}}\right)+\frac{1}{\sqrt{3}}\mathcal{Z}$
to the quantum case.  

\begin{lem}
 If the convolutional strong subadditivity holds, then 
 the triangle inequality of quantum Ruzsa divergence also holds.
 
\end{lem}
\begin{proof}
    This comes from the entropy inequality for quantum convolution 
    $ S(\rho\boxtimes\tau\boxtimes\sigma)\geq \max\set{ S(\rho\boxtimes\tau)\;, S(\sigma)}$ \cite{BGJ23a,BGJ23b}.
\end{proof}

Here, let us show that the convolutional strong subadditivity holds with some additional 
assumption on the states.
\begin{prop}
The convolutional strong subadditivity holds for 
the following two cases,

(1) the quantum states
$\rho, \sigma$, and $ \tau$ are all stabilizer states; 

(2) the quantum states
$\rho, \sigma$, and $ \tau$ are diagonal states in the computational basis.

\end{prop}
\begin{proof}
(1) 
Let us assume that  quantum states
$\rho, \sigma$, and $ \tau$ are all stabilizer states with 
stabilizer group $G_{\rho}$, $G_{\sigma}$ and $G_{\tau}$. Then 
the absolute value of the characteristic $\Xi_{\rho}$ is either equal to 
$1$ or $0$, and $\Xi_{\rho}(\vec x)$ equals  $1$ iff the Weyl operator $w(\vec x)$ belongs 
to the stabilizer group. Similar arguments also work for $\sigma$ and $\tau$.

Since 
\begin{align*}
    \Xi_{\rho\boxtimes\tau\boxtimes\sigma}(\vec x)
    =\Xi_{\rho}(ls\vec x)\,\Xi_{\tau}(lt\vec x)\,
    \Xi_{\sigma}(m\vec x)\;,
\end{align*}
with $s^2+t^2\equiv 1\mod d, l^2+m^2\equiv 1\mod d$, 
we have  $|\Xi_{\rho\boxtimes\tau\boxtimes\sigma}(\vec x)|$ is equal to $1$ or $0$, and 
\begin{eqnarray*}
    |\Xi_{\rho\boxtimes\tau\boxtimes\sigma}(\vec x)|=1\;,~~
    \text{iff}~~w(\vec{x})\in G_{\rho}\cap G_{\sigma}\cap G_{\sigma}\;.
\end{eqnarray*}
That is, $\rho\boxtimes \tau\boxtimes\sigma$ is a stabilizer state with 
$G_{\rho}\cap G_{\sigma}\cap G_{\tau}$ as the stabilizer group, 
and the  quantum entropy is
\begin{eqnarray*}
    S(\rho\boxtimes \tau\boxtimes\sigma)
    =\log\frac{d^n}{|G_{\rho}\cap G_{\sigma}\cap G_{\tau}|}\;.
\end{eqnarray*}

Using the same reasoning, we can also prove 
that $\rho\boxtimes\sigma$ is a stabilizer state 
with $G_{\rho}\cap G_{\sigma}$ as its stabilizer group, 
and the quantum entropy is 
\begin{eqnarray*}
        S(\rho\boxtimes \sigma)
    =\log\frac{d^n}{|G_{\rho}\cap G_{\sigma}|}\;.
\end{eqnarray*}
Similarly, $\sigma\boxtimes\tau$
is a stabilizer state with $G_{\sigma}\cap G_{\tau}$ as its stabilizer group, and the 
quantum entropy is 
\begin{eqnarray*}
        S(\sigma\boxtimes \tau)
    =\log\frac{d^n}{|G_{\sigma}\cap G_{\tau}|}\;.
\end{eqnarray*}
Besides, the quantum entropy of the stabilizer state $\sigma$ is 
\begin{eqnarray*}
    S(\sigma)=\log\frac{d^n}{|G_{\sigma}|}.
\end{eqnarray*}

Hence, to prove the convolutional strong subadditivity, it is equivalent to proving the following statement, 
\begin{eqnarray*}
    \frac{1}{|G_{\rho}\cap G_{\tau}\cap G_{\sigma}||G_{\sigma}|}
    \leq  \frac{1}{|G_{\rho}\cap G_{\sigma}||G_{\sigma}\cap G_{\tau}|}\;.
\end{eqnarray*}
That is 
\begin{eqnarray*}
    \frac{|G_{\rho}\cap G_{\sigma}|}{|G_{\rho}\cap G_{\sigma}\cap G_{\tau}|}\leq \frac{|G_{\sigma}|}{|G_{\sigma}\cap G_{\tau}|}\;.
\end{eqnarray*}

Since $G_{\rho}, G_{\sigma}, G_{\tau}$ are all finite abelian group, then 
the above statement is equivalent to the 
\begin{eqnarray*}
    |G_{\rho}\cap G_{\sigma}:G_{\rho}\cap G_{\sigma}\cap G_{\tau}|
    \leq |G_{\sigma}:G_{\sigma}\cap G_{\tau}|\;,
\end{eqnarray*}
where the $[G:H]$ is the index of a subgroup H in a group G.
This comes from the following property of the index of groups:
if $H,K$ are subgroups of $G$, then 
\begin{eqnarray*}
    |H:H\cap K|
    \leq |G:K|.
\end{eqnarray*}

(2) 
Let us consider that $2s^2\equiv 1\mod d$, $m\equiv ls$, and $l^2+m^2\equiv 1\mod d$.
Since $\rho,\sigma, \tau$ are diagonal in the computational basis, they can be written as follows
\begin{align}
    \rho=\sum_{\vec x\in \mathbb{Z}^n_d}p(\vec x)\proj{\vec x},~
    \sigma=\sum_{\vec x\in \mathbb{Z}^n_d}q(\vec x)\proj{\vec x},~
    \tau=\sum_{\vec x\in \mathbb{Z}^n_d}r(\vec x)\proj{\vec x}\;,
\end{align}
where $\set{p(\vec x)},\set{q(\vec x)},\set{r(\vec x)}$ are probability distributions 
on $\mathbb{Z}^n_d$.

Let us consider three random variables $\mathcal{X}, \mathcal{Y}, \mathcal{Z}$, which takes values in $\mathbb{Z}^n_d$ as 
follows
\begin{align*}
    \text{Pr}[\mathcal{X}=\vec x]=p(\vec x)\;, \quad
    \text{Pr}[\mathcal{Y}=\vec x]=q(\vec x)\;, \quad
    \text{Pr}[\mathcal{Z}=\vec x]=r(\vec x).
\end{align*}
Then, after some calculation, we find that 
\begin{align*}
    S(\rho\boxtimes\tau\boxtimes\sigma)=H(\mathcal{X}+\mathcal{Y}+\mathcal{Z})\;,\quad S(\sigma)=H(\mathcal{Y});\\
    S(\rho\boxtimes\sigma)=H(\mathcal{X}+\mathcal{Y})\;,\quad S(\sigma\boxtimes\tau)=H(\mathcal{Y}+\mathcal{Z})\;,
\end{align*}
where $H(\mathcal{X})$ is the Shannon entropy of the discrete random variable $\mathcal{X}$, and 
$\mathcal{X}+\mathcal{Y}$ is the sum of random variables $\mod d$.
Hence, the convolutional strong-subadditivity in this case is reduced to 
the classical case
\begin{align}
    H(\mathcal{X}+\mathcal{Y}+\mathcal{Z})+H(\mathcal{Y})
    \leq H(\mathcal{X}+\mathcal{Y})+H(\mathcal{Y}+\mathcal{Z}),
\end{align}
which is true.

\end{proof}

Note that, in the convolutional strong subadditivity~\eqref{ineq:CSSA}, we consider the quantum convolution of three states as 
$(\rho\boxtimes\tau)\boxtimes\sigma$, where $s^2+t^2\equiv 1\mod d$, $s\equiv t\mod d$, $m\equiv ls \mod d$, and $l^2+m^2\equiv 1\mod d$.
We may also consider other possible quantum convolutions on three input states, like the one in~\eqref{eq:qub_con} defined 
on an $n$-qubit system.

Moreover, the usual strong subadditivity in quantum information theory is: given  a tripartite state $\rho_{ABC}$, 
    it holds that $S(\rho_{ABC})+S(\rho_C)\leq S(\rho_{AC})+S(\rho_{BC})$, which was conjectured by Robinson, Ruelle~\cite{RobinsonCMP67} and  Lanford, Robinson~\cite{Lanford68} and later proved by Lieb and Ruskai~\cite{LiebJMP73}.
  Our inequality~\eqref{ineq:CSSA}
    shares a similar form, so that we call 
    \eqref{ineq:CSSA} as "Convolutional strong subadditivity". Besides the strong subadditivity, there is a subadditivity for bipartite states, i.e., 
    $S(\rho_{AB})\leq S(\rho_A)+S(\rho_B)$. It is natural 
    to consider the convolutional subadditivity or supadditivity. However, neither of them holds, as 
    we can give some counterexamples in the following theorem.

\begin{thm}[\bf No subadditivity or supadditivity for the quantum convolution]\label{thm:sub_con}
(1) There exist quantum states $\rho$ and $\sigma$ such that 
\begin{eqnarray}
    S(\rho\boxtimes\sigma)>S(\rho)+S(\sigma).
\end{eqnarray}

(2) There also exist quantum states $\rho, \sigma$ such that 
\begin{eqnarray}
    S(\rho\boxtimes\sigma)<S(\rho)+S(\sigma).
\end{eqnarray}

\end{thm}
\begin{proof}
(1)
Let us take $\rho$ to be the eigenstate of Pauli $Z$ operator corresponding to $+1$ eigenvalue, i.e., 
\begin{eqnarray*}
    \rho=\frac{1}{d^n}\sum_{\vec a\in \mathbb{Z}^n_d} Z^{\vec a}\;,
\end{eqnarray*}
and $\sigma $ to be the eigenstate of Pauli $X$ operator corresponding to $+1$ eigenvalue, i.e., 
\begin{eqnarray*}
    \sigma=\frac{1}{d^n}\sum_{\vec a\in \mathbb{Z}^n_d} X^{\vec a}\;.
\end{eqnarray*}
Then $S(\rho)=S(\sigma)=0$ and $\rho,\sigma$ are not commuting with each other. However, 
$\rho\boxtimes\sigma=\frac{I}{d^n}$, as 
\begin{eqnarray*}
    \Xi_{\rho\boxtimes\sigma}(\vec x)
    =\Xi_{\rho}(s\vec x)\Xi_{\sigma}(t\vec x)
    =0, 
\end{eqnarray*}
for any $\vec x\neq 0$. 
Hence $$S(\rho\boxtimes\sigma)=n\log d>S(\rho)+S(\sigma)=0.$$

(2) 
Let us take both $\rho$ and $\sigma$ to be the maximally mixed state $I/d^n$, then 
$$S(\rho)=S(\sigma)=n\log d.$$
Moreover, $\rho\boxtimes\sigma=I/d$, then 
$$S(\rho\boxtimes\sigma)=n\log d<S(\rho)+S(\sigma)=2n\log d.$$
\end{proof}

In the classical setting, convolution satisfies a fundamental inequality: 
$H(\mathcal{X}+\mathcal{Y})\leq H(\mathcal{X})+H(\mathcal{Y})$~\cite{tao2006additive,cover1999elements}. However, in the quantum domain, such inequality is not generally valid, as demonstrated in Theorem~\ref{thm:sub_con}.  This illustrates a fundamental discrepancy between the classical and quantum settings regarding convolutional inequalities.

Usually, the discrepancy 
between classical and quantum settings stems from
the quantum feature. For example, 
in classical information theory, the conditional 
entropy $H(\mathcal{X}|\mathcal{Y})$ is nonnegative. 
However, in quantum information theory, the conditional quantum entropy 
can be negative ~\cite{CerfPRL97,horodecki2005partial},
a phenomenon that can detect the presence of quantum entanglement and offer advantages in quantum communication.
In this work, we 
 explore how this discrepancy in convolutional subadditivity could be utilized to detect magic in quantum states.

\begin{prop}[Violation of convolutional subadditivity implies magic]
 If $\rho$ is an $n$-qudit MSPS, then subadditivity holds, that is
 \begin{align}
     S(\rho\boxtimes\rho)\leq 2S(\rho).
 \end{align}
Hence, for pure state $\rho$,  the violation of convolutional subadditivity $\rho$  implies that $\rho$  is a magic state. 
 \end{prop}
\begin{proof}
 If $\rho$ is an MSPS, then $\rho=\mathcal{M}(\rho)$, which is fixed under quantum convolution $\boxtimes$ up to some Weyl operator. That is, there exists some Weyl operator $w(\vec x)$ such that
    $\rho\boxtimes \rho =w(\vec x)\rho w(\vec x)^\dag$. 
    Hence, $ S(\rho\boxtimes\rho)=S(w(\vec x)\rho w(\vec x)^\dag)=S(\rho)\leq 2S(\rho)$. Hence, the convolutional subadditivity holds in this case.

Moreover, if $\rho$ is a pure state,  then $S(\rho)=0$. 
If the pure state $\rho$ is a stabilizer state, then $\rho\boxtimes \rho$
is also a pure stabilizer state. 
Then the violation of convolutional subadditivity means $\rho\boxtimes\rho$ is a mixed state. Therefore, the pure state $\rho$ is nonstabilizer, i.e., a magic state.

\end{proof}

\begin{Rem}
Note that the violation of convolutional subadditivity is also closely related to the negative conditional quantum entropy.
Consider the bipartite state obtained by applying a convolutional unitary on the tensor product of two input states $\rho$, 
 i.e.,  $\rho_{AB}=U_{s,t}(\rho\ot \rho) U^\dag_{s,t}$. Then, the reduced state $\rho_A$ on the subsystem $A$ is
$\rho_A=\Ptr{B}{\rho_{AB}}=\Ptr{B}{U_{s,t}(\rho\ot \rho)U^\dag_{s,t}}=\rho\boxtimes\rho$. Hence, the conditional quantum entropy $S(B|A)$ is
\begin{align}
    \nonumber S(B|A)=&S(\rho_{AB})-S(\rho_A)=S(U_{s,t}(\rho\ot \rho) U^\dag_{s,t})-S(\rho\boxtimes\rho)\\
    =&2S(\rho)-S(\rho\boxtimes\rho).
\end{align}
Therefore,  a negative conditional quantum entropy is equivalent to the violation of the convolutional subadditivity. 
This equivalence indicates that we are using the quantum entanglement to detect the magic of states by employing the convolutional unitary to entangle the independent input states.

\end{Rem}

\subsection{Quantum-doubling constant and quantum inverse sumset theorem}\label{sec:QIST}
In this section, we use the quantum Ruzsa divergence 
to introduce a quantum-doubling constant and 
quantum inverse sumset theorem to further explore the stabilizer structure 
of quantum states. Before that, we first recall the classical inverse sumset theorem from Tao's work~\cite{tao_2010}, where some quantity, called doubling constant, 
was introduced to quantify the distance between 
the given random variable and the uniform distribution on the (cosets of) finite  subgroups.
\begin{lem}[Theorem 1.1 in \cite{tao_2010}]
Let $G$ be an additive group, $\mathcal{X}$ be a random variable taking value from $G$, 
and doubling constant $\delta(\mathcal{X})=\exp(H(\mathcal{X}+\mathcal{X}')-H(\mathcal{X}))$ with 
$H(\cdot)$ being the Shannon entropy. Then we have\\
(1) $\delta(\mathcal{X})=1$ if and only if $\mathcal{X}$ is the uniform distribution on a coset of a finite
subgroup of $G$;\\
(2)If $\delta(\mathcal{X})\leq C$ where $C$ is a constant, then there exists a coset progression $H+P$ of rank $O_K(1)$
such that distance between $\mathcal{X}$ and  the uniform distribution on $H+P$ is $\ll_K1$.
\end{lem}

Inspired by the classical work, we introduce a quantum version of doubling constant based on our quantum convolution.

\begin{Def}[\bf Quantum-doubling constant]\label{def:QDC}
Given an $n$-qudit state $\rho$, the quantum-doubling constant is 
\begin{eqnarray}
    \delta_q[\rho]
    =\exp(S(\rho\boxtimes\rho)-S(\rho)),
\end{eqnarray}
where $\boxtimes$ denotes the quantum convolution in Definition~\ref{def:conv}.
In general, the $\alpha$-order quantum-doubling constant is 
\begin{eqnarray}
\delta_{q,\alpha}[\rho]
    =\exp(S_{\alpha}(\rho\boxtimes\rho)-S_{\alpha}(\rho)).
\end{eqnarray}
\end{Def}

Based on the definition, the quantum-doubling constant 
is equal to $\exp(D_{Rz}(\rho||\rho))$, which is also 
the entropy difference of the first step in the
q-CLT. Moreover, 
for  a pure state $\psi$, the quantum-doubling constant 
$\delta_q[\psi]$ is equal to the magic entropy 
$ME(\psi)=S(\psi\boxtimes\psi)$  defined in \cite{BGJ23c} up 
to a logarithm.

\begin{Rem}
    Similar to the quantum-doubling constant, we can also 
    define the quantum-difference constant is 
\begin{eqnarray}
        \delta^{-}_q[\rho]
    =\exp(S(\rho\boxminus\rho)-S(\rho))\;,
\end{eqnarray}
where $\rho\boxminus\rho=\Ptr{A}{U_{s,t}(\rho\ot\rho) U^\dag_{s,t}}$, i.e., the complementary channel 
of the quantum channel $\boxtimes$.
\end{Rem}

In this work, we consider the following problem:
given a quantum state $\rho$, how could the quantum-doubling constant $\delta_q[\rho]$ tell the structure of the 
state $\rho$, that is, how close the state is to the set of MSPSs.
We call this the quantum inverse sumset problem. Here, we focus on the pure state case, for which we have the following result.

\begin{thm}[\bf Quantum inverse sumset theorem using magic gap]\label{thm:QIST}
     Given an $n$-qudit pure state $\psi$, 

    (1) $\delta_q[\psi]\geq 1$, with equality iff $\psi\in \text{STAB}$. 

    (2) If $1<\delta_q[\psi]\leq C$, 
    then 
    \begin{eqnarray}\label{240929eq2}
        D(\psi||\mathcal{M}(\psi))\leq \frac{\log R(\psi)}{\log R(\psi)- \log[1+\lambda(R(\psi)-1)]}\log C,
    \end{eqnarray}
where $\lambda= (1-MG(\rho))^2$.
\end{thm}
\begin{proof}
    (1) It comes directly from the entropy inequality  for quantum convolution in \cite{BGJ23a,BGJ23b} (see Proposition 54 and Corollary 59  in \cite{BGJ23b}).

    (2) First, 
  by the monotonicity of relative entropy under the quantum channel, there exists some 
    factor $\kappa_{\psi}\leq 1$ such that
    \begin{eqnarray*}
        D(\psi\boxtimes\psi||\mathcal{M}(\psi)\boxtimes\psi)
        =D(\psi||\mathcal{M}(\psi))\kappa_{\psi}.
    \end{eqnarray*}
By  Lemma~\ref{lem:ent_equ} and the fact that $S(\mathcal{M}(\psi)\boxtimes\psi)=S(\mathcal{M}(\psi))$, it can 
be rewritten as 
    \begin{eqnarray*}
       S(\mathcal{M}(\psi))-S(\psi\boxtimes\psi)
        =[S(\mathcal{M}(\psi))-S(\psi)]\kappa_{\psi}.
    \end{eqnarray*}
    Since $\delta_q[\psi]>1$, $ \psi$ is not a stabilizer state, and thus $\kappa_{\psi}<1$. 
    Hence 
    \begin{eqnarray*}
        S(\psi\boxtimes\psi)-S(\psi)=(1-\kappa_{\psi})[S(\mathcal{M}(\psi))-S(\psi)]\;,
    \end{eqnarray*}
    which, by Theorem 17 in \cite{BGJ23b}, implies that 
    \begin{eqnarray*}
        D(\psi||\mathcal{M}(\psi))=
        \frac{1}{1-\kappa_{\psi}}\left[S(\psi\boxtimes\psi)-S(\psi)\right]
        \leq \frac{1}{1-\kappa_{\psi}}\log C.
    \end{eqnarray*}
Now, let us provide an upper bound on the factor $\kappa_{\psi}$ using magic gap.
    \begin{align*}
    \kappa_{\psi}
    =&\frac{D(\psi\boxtimes\psi||\mathcal{M}(\psi))}{D(\psi||\mathcal{M}(\psi))}
    =\frac{S(\mathcal{M}(\psi))-S(\psi\boxtimes\psi)}{S(\mathcal{M}(\psi))}\\
    \leq& \frac{S(\mathcal{M}(\psi))-S_2(\psi\boxtimes\psi)}{S(\mathcal{M}(\psi))},
\end{align*}
where the inequality comes from the fact that $S_{\alpha}$ is nonincreasing with respect to $\alpha$. 

Let us assume that $G_{\psi}$ is the stabilizer group of $\psi$, then 
$R(\psi)=\frac{d^n}{|G_{\psi}|}$, and thus $S(\mathcal{M}(\psi))=\log R(\psi)=\log\frac{d^n}{|G_{\psi}|} $. 
Hence
\begin{align*}
    &S(\mathcal{M}(\psi))-S_2(\psi\boxtimes\psi)\\
    =&\log \frac{d^n}{|G_{\psi}|}
    +\log \left(\frac{|G_{\psi}|}{d^n}+\frac{1}{d^n}\sum_{\vec x \notin G}|\Xi_{\psi}(s\vec x)|^2 |\Xi_{\psi}(t\vec x)|^2\right)\\
    \leq& \log \frac{d^n}{|G_{\psi}|}
    +\log \left(\frac{|G_{\psi}|}{d^n}+\frac{\lambda}{d^n}\sum_{\vec x \notin G}|\Xi_{\psi}(\vec x)|^2\right)\\
    =& \log \frac{d^n}{|G_{\psi}|}
    +\log \left[\frac{|G_{\psi}|}{d^n}+\lambda\left(1-\frac{|G_{\psi}|}{d^n}\right)\right]\\
    =&\log\left[1+\lambda\left(\frac{d^n}{|G_{\psi}|}-1\right)
    \right],
\end{align*}
where the third line comes from the definition of $\lambda$, and the fourth line comes from the fact that 
\begin{align}
       1= \Tr{\psi^2}
   =\frac{1}{d^n}\sum_{\vec x}|\Xi_{\psi}(\vec x)|^2
   =\frac{|G_{\psi}|}{d^n}+\frac{1}{d^n}\sum_{\vec x\notin G_{\psi}}|\Xi_{\psi}(\vec x)|^2.
\end{align}

\end{proof}

\begin{Rem}
In the aforementioned results, we explore the inverse sumset theorem for pure states. When considering mixed states $\rho$,
the inequality in (1) 
of Theorem \ref{thm:QIST} remains valid,
which takes equality if and only if $\rho$ is an MSPS (see the property (1) in Proposition~\ref{thm:QRD} ).
In addition, we also expect that the property (2) in Theorem \ref{thm:QIST} holds, which  has the following form
\begin{align}\nonumber&D(\rho||\mathcal{M}(\rho))\\
\label{eq:new}\leq& \frac{\log R(\rho)}{\log R(\rho)- \log[1+\lambda(R(\rho)\trace[\rho^2]-1)]}\log C.
\end{align}

However,
the method 
to prove the property (2) 
in Theorem \ref{thm:QIST} 
only yields the following inequality:
\begin{align}
  \nonumber  &D(\rho||\mathcal{M}(\rho))\\
    \leq& \frac{\log R(\rho)-S(\rho)}{\log R(\rho)- \log[1+\lambda(R(\rho)\trace[\rho^2]-1)]-S(\rho)}\log C,
\end{align}
when $\log R(\rho)- \log[1+\lambda(R(\rho)\trace[\rho^2]-1)]-S(\rho)>0$.  This inequality is weaker than \eqref{eq:new} that we expected, as
$\frac{x}{y}\leq \frac{x-a}{y-a}$ for $x\geq y>a\geq 0$.
Therefore, the above method 
cannot provide a good estimate  for mixed states, which may require new techniques. We leave it for a future study.
\end{Rem}

Now, let us consider the properties of quantum-doubling constant.
\begin{cor}\label{prop:QDC}
Given an $n$-qudit state $\rho$, the quantum-doubling constant satisfies the following properties:

(1) {\bf Positivity:}  $\delta_q[\rho]\geq 1$, with equality iff $\rho\in \text{MSPS}$. 

    (2) {\bf Additivity under tensor product:} $\delta_q[\rho_1\ot\rho_2]=\delta_q[\rho_1]\delta_q[\rho_2]$.

    (3) {\bf Invariance under Clifford unitary:} $\delta_q\left[U\rho U^\dag\right]=\delta_q[\rho]$ for any Clifford unitary $U$.

    (4) {\bf Monotonicity under partial trace:} $\delta_q\left[\Ptr{i}{\rho}\right]\leq \delta_q[\rho]$, where $\Ptr{i}{\cdot}$ denotes the 
    partial trace on the $i$-th qudit for any $i\in [n]$.
    
\end{cor}

\begin{proof}
    These results come directly from the properties of quantum Ruzsa divergence in Theorem~\ref{thm:QRD}.
\end{proof}

The above properties suggest that the quantum-doubling constant can serve as a measure of magic. Moreover, compared to the 
quantum Ruzsa divergence of magic $M_{Rz}$ in Definition \ref{Def:MRZ}, which requires the minimization of all stabilizer states, the quantum-doubling constant does not require such optimization.  The absence of optimization will make it easier to 
compute than $M_{Rz}$.

In general, it is hard to compare the quantum Rusza divergence of magic $M_{Rz}$  and quantum-doubling constant $\delta_q$, as  $M_{Rz}$ involves the optimization over all stabilizer states, and $D_{Rz}(\rho||\sigma)$ is not symmetric 
with respect to $\rho$ and $\sigma$. However, we can get some nice relationship between  $M_{Rz}$ and $\delta_q$
under certain conditions and assumptions, as shown in the following result. 

\begin{prop}
 Consider the quantum convolution $\boxtimes$ defined by the balanced beam splitter, i.e.,  $s\equiv t\mod d$, and assume that 
the Conjecture \ref{Conj:sub_RZ} holds. Then we have the following relationship for any $n$-qudit pure state $\psi$, 
\begin{align}
    \log \delta_q(\psi)\leq 2M_{Rz}(\psi).
\end{align}
\end{prop}
\begin{proof}
First, by the concavity of quantum entropy, the minimization in $ \min_{\sigma\in\text{STAB}}S(\rho\boxtimes\sigma)-S(\rho)$
is taken over pure stabilizer states. Hence, for pure state $\psi$, 
$M_{Rz}(\psi)=\min_{\ket{\phi}\in STAB}S(\psi\boxtimes\phi)$. 
Moreover, since the quantum convolution $\boxtimes$ is defined via the balanced beam splitter,
then 
\begin{align}
    \psi\boxtimes\phi=\phi\boxtimes\psi.
\end{align}
Because
the characteristic functions $ \Xi_{\psi\boxtimes\phi}$ and $ \Xi_{\phi\boxtimes\psi}$ are equal as
\begin{align*}
    \Xi_{\psi\boxtimes\phi}(\vec x)=\Xi_{\psi}(s\vec x)\Xi_{\phi}(s\vec x),\\
        \Xi_{\phi\boxtimes\psi}(\vec x)=\Xi_{\phi}(s\vec x)\Xi_{\psi}(s\vec x).
\end{align*}
Hence, there exists a pure stabilizer state $\phi_0$ such that 
\begin{align}
    M_{Rz}(\psi)= S(\psi\boxtimes\phi_0)=S(\phi_0\boxtimes\psi).
\end{align}

Moreover, by applying \eqref{ineq:tri_equiv} in Conjecture \ref{Conj:sub_RZ} 
with  $\rho=\psi, \tau=\psi, \sigma=\phi_0$,
we have 
\begin{align}
    S(\psi\boxtimes\psi)
    \leq S(\psi\boxtimes\phi_0)+S(\phi_0\boxtimes\psi)=2M_{Rz}(\psi).
\end{align}
This completes the proof.

\end{proof}

Note that the quantum-doubling constant is defined for qudit systems. To extend this concept to the qubit case,
we  introduce "quantum tripling constant" by using the quantum convolution $\boxtimes_3$ on three input states, i.e., choosing
$K=3$ in the Definition~\ref{Def:conv_qubit}.

% The above properties indicate that the quantum doubling constant can be used as a magic measure. Moreover, 
%  we can generalize the quantum doubling constant to the qubit case, which is called
%  the quantum tripling constant as follows.
\begin{Def}[\bf Quantum tripling constant]
Given an $n$-qubit state $\rho$, the quantum tripling constant is 
\begin{eqnarray}
    \tilde{\delta}_q[\rho]
    =S(\boxtimes_3\rho)-S(\rho),
\end{eqnarray}
where the quantum convolution $\boxtimes_3$ is defined in \eqref{eq:qub_con}.
\end{Def}

Similar to the qudit case, we also have the following result for the $n$-qubit pure state.

\begin{prop}[\bf Quantum inverse sumset theorem for qubits]
     Given an $n$-qubit pure state $\psi$, 

    (1) $ \tilde{\delta}_q[\psi]\geq 1$, with equality iff $\psi\in \text{STAB}$. 

    (2) If $1< \tilde{\delta}_q[\psi]\leq C$, 
    then 
    \begin{align}
        D(\psi||\mathcal{M}(\psi))\leq \frac{\log R(\psi)}{\log R(\psi)- \log[1+\lambda^2(R(\psi)-1)]}\log C,
    \end{align}
    where $\lambda= (1-MG(\rho))^2$.

\end{prop}
\begin{proof}
    The proof is similar to that of the qudit case.
\end{proof}

Note that, the properties of the quantum-doubling constant in qudits in Proposition \ref{prop:QDC} also hold for 
the quantum-tripling constant in qubits by using the properties of $\boxtimes_3$ in \cite{BGJ23c}.

\section{Conclusion}

In this work, we have introduced the 
quantum Ruzsa divergence
to study the stabilizer structure of quantum states.
By using quantum Ruzsa divergence, 
we propose two  new magic measures, quantum Ruzsa divergence of magic and 
quantum-doubling constant,  to quantify the amount of magic in quantum states.
 We also pose and study an interesting conjecture called ``convolutional strong subadditivity''.

There are still many interesting problems to solve within the current framework. We list some of them here:

(1) We have proven convolutional strong subadditivity for two specific cases: when all the input states are either diagonal or stabilizer states. Can this inequality be extended to hold for any input states? Moreover,
can one generalize other classical sumset and inverse sumset results to 
our quantum convolutional framework? For example, what is the quantum version of the polynomial Freiman–Ruzsa conjecture in this setting? 
The  interesting recent work of Gowers and collaborators~\cite{gowers2023conjecture} may be helpful.

(2)
One can also generalize the quantum Ruzsa divergence to bosonic quantum systems by using the CV beam splitter as the quantum convolution. Define the 
bosonic Ruzsa divergence as 
$D^B_{Rz}(\rho||\sigma)=S(\rho\boxplus_\lambda\sigma)-S(\rho)$  in a similar way as Definition~\ref{Def:QRD}
and the bosonic quantum-doubling constant as $\delta_{B}(\rho)=S(\rho\boxplus_\lambda\rho)-S(\rho)$ in a similar way as Definition~\ref{def:QDC}, where $\boxplus_\lambda$ is
the bosonic convolution defined via beam splitter~\cite{Konig13,Konig14,Palma14}. Then we can use the bosonic
Ruzsa divergence to 
 design new non-Gaussian measures in the resource theory of non-Gaussianity~\cite{TakagiPRA18,ZhuangPRA18}. We plan to develop these ideas within the  CV setting in  future work. 

(3) One important application of quantum Ruzsa divergence in this work is the new design of magic measures to quantify the amount of magic of states. Can we find more physical interpretations and applications of these
mathematical results?
\section*{Acknowledgment}

The authors would like to thank Michael Freedman, Yichen Hu, Bryna Kra, Yves Hon Kwan, Xiang Li, Elliott Lieb, Freddie Manners, Graeme Smith, and Yufei Zhao for the helpful discussion.
This work was supported in part by the ARO Grant W911NF-19-1-0302 and the ARO
MURI Grant W911NF-20-1-0082.

\ifCLASSOPTIONcaptionsoff
  \newpage
\fi

\begin{IEEEbiographynophoto}{Kaifeng Bu}
 Kaifeng Bu got his B.S. and Ph.D. from Zhejiang University in 2014 and 2019, respectively, and was a postdoctoral researcher at Harvard University from 
 2019 to 2024. He is now an assistant professor in the Department of Mathematics at Ohio State University. His research focuses on the advantages of quantum computing and quantum information processing, as well as on the interplay of quantum information with  computer science, physics, and mathematics.
\end{IEEEbiographynophoto}

\begin{IEEEbiographynophoto}{Weichen Gu}
Weichen Gu received a B.S. degree in mathematics from the University of Science and Technology of China in 2014, an M.S. degree in mathematics from the University of Chinese Academy of Sciences in 2017,
and a Ph.D. degree in mathematics from the University of New Hampshire in 2024.
He is now a postdoctoral researcher  in the Department of Mathematics at Ohio State University. His research interests include quantum information theory, quantum computing, operator algebra, number theory, and combinatorics.
\end{IEEEbiographynophoto}

\begin{IEEEbiographynophoto}{Arthur Jaffe}
Arthur Jaffe is the Landon T. Clay Professor of Mathematics and Theoretical Science
at Harvard University. His past work on physics and mathematics includes giving the first mathematical examples that combine special relativity, quantum theory, and interaction.  His recent research focuses on quantum information.  He was a founder and first president of the Clay Mathematics Institute,  president of the International Association of Mathematical Physics,  and president of
the American Mathematical Society.  He is a Fellow of the Hagler Research Institute at Texas A\&M University. 
He is a member of the U.S. National Academy of Sciences, a fellow of the American Academy of Arts and Sciences, and an Honorary Member of the Royal Irish Academy. 
\end{IEEEbiographynophoto}

%\newpage



%\vfill


\end{document}